\documentclass[conference]{IEEEtran}
\IEEEoverridecommandlockouts
\usepackage{cite}
\usepackage{amsmath,amssymb,amsfonts}
\usepackage[ruled,linesnumbered,vlined]{algorithm2e}
\usepackage{algorithmic}
\usepackage{graphicx}
\usepackage{textcomp}
\usepackage{xcolor}
\usepackage{enumerate}
\usepackage{ntheorem}
\usepackage{subfigure}
\usepackage{balance}
\usepackage{bm}
\usepackage{hyperref}
\newtheorem{definition}{Definition}
\newtheorem{lemma}{Lemma}
\newtheorem{example}{Example}
\newcommand{\nop}[1]{}

\newtheorem*{proof}{Proof:}
\def\BibTeX{{\rm B\kern-.05em{\sc i\kern-.025em b}\kern-.08em
    T\kern-.1667em\lower.7ex\hbox{E}\kern-.125emX}}
\begin{document}

\makeatletter
\newcommand{\linebreakand}{%
  \end{@IEEEauthorhalign}
  \hfill\mbox{}\par
  \mbox{}\hfill\begin{@IEEEauthorhalign}
}
\makeatother

\title{Reverse Influential Community Search Over Social Networks (Technical Report)}

\author{
\IEEEauthorblockN{Qi Wen}
\IEEEauthorblockA{\textit{Software Engineering Institute} \\
\textit{East China Normal University}\\
Shanghai, China \\
51265902057@stu.ecnu.edu.cn}
\and
\IEEEauthorblockN{Nan Zhang}
\IEEEauthorblockA{\textit{Software Engineering Institute} \\
\textit{East China Normal University}\\
Shanghai, China \\ 
51255902058@stu.ecnu.edu.cn}
\and
\IEEEauthorblockN{Yutong Ye}
\IEEEauthorblockA{\textit{Software Engineering Institute} \\
\textit{East China Normal University}\\
Shanghai, China \\ 
52205902007@stu.ecnu.edu.cn}
\linebreakand
\IEEEauthorblockN{Xiang Lian}
\IEEEauthorblockA{\textit{Department of Computer Science} \\
\textit{Kent State University}\\
Kent, OH 44242, USA \\
xlian@kent.edu}
\and
\IEEEauthorblockN{Mingsong Chen}
\IEEEauthorblockA{\textit{Software Engineering Institute} \\
\textit{East China Normal University}\\
Shanghai, China \\ 
mschen@sei.ecnu.edu.cn}
}

\maketitle

\begin{abstract}
As an important fundamental task of numerous real-world applications such as social network analysis and online advertising/marketing, several prior works studied influential community search, which retrieves a community with high structural cohesiveness and maximum influences on other users in social networks. However, previous works usually considered the influences of the community on \textit{arbitrary} users in social networks, rather than specific groups (e.g., customer groups, or senior communities). Inspired by this, we propose a novel \textit{\textbf{Top-M R}everse \textbf{I}nfluential \textbf{C}ommunity \textbf{S}earch} (Top\textit{M}-RICS) problem, which obtains a \textit{seed community} with the maximum influence on a user-specified \textit{target community}, satisfying both structural and keyword constraints. To efficiently tackle the Top\textit{M}-RICS problem, we design effective pruning strategies to filter out false alarms of candidate seed communities, and propose an effective index mechanism to facilitate the community retrieval. 
We also formulate and tackle a Top\textit{M}-RICS variant, named \textit{\textbf{Top-M R}elaxed \textbf{R}everse \textbf{I}nfluential \textbf{C}ommunity \textbf{S}earch} (Top\textit{M}-R$^2$ICS), 
which returns top-\textit{M} subgraphs with relaxed structural constraints and having the maximum influence on a user-specified target community.
Comprehensive experiments have been conducted to verify the efficiency and effectiveness of our Top\textit{M}-RICS and Top\textit{M}-R$^2$ICS approaches on both real-world and synthetic social networks under various parameter settings.
\end{abstract}


\section{Introduction}
For the past decades, the \textit{community search} has attracted much attention in various real-world applications such as online advertising/marketing \cite{tu2022viral,ebrahimi2022social,rai2023top, MolaeiFB24}, social network analysis \cite{kumar2022influence,subramani2023gradient,li2015influential, YanLWDW20}, and many others. Prior works on the community search \cite{zhou2023influential,islam2022keyword,wu2021efficient,al2020topic} usually retrieved a community (subgraph) of users from social networks with high structural and/or spatial cohesiveness. Several existing works \cite{wu2021efficient,li2022itc,xu2020personalized} considered the influences of communities and studied the problem of finding communities with high influences on other users in social networks.

In this paper, we propose a novel problem, named \textit{Top-M Reverse Influential Community Search} (Top\textit{M}-RICS) over social networks, which obtains \textit{M} communities (w.r.t. specific interests such as sports, food, etc.) that have high structural cohesiveness and the highest influences on a targeted group (community) of users (instead of arbitrary users in social networks). The resulting Top\textit{M}-RICS communities are useful for various real applications such as online advertising/marketing in social media \cite{fang2014topic} and disease spread prevention in contact networks \cite{firestone2011importance}. Below, we give motivation examples of our Top\textit{M}-RICS problem.

\begin{figure}[t!]\vspace{-3ex}
    \centering
    \includegraphics[width=0.5\textwidth]{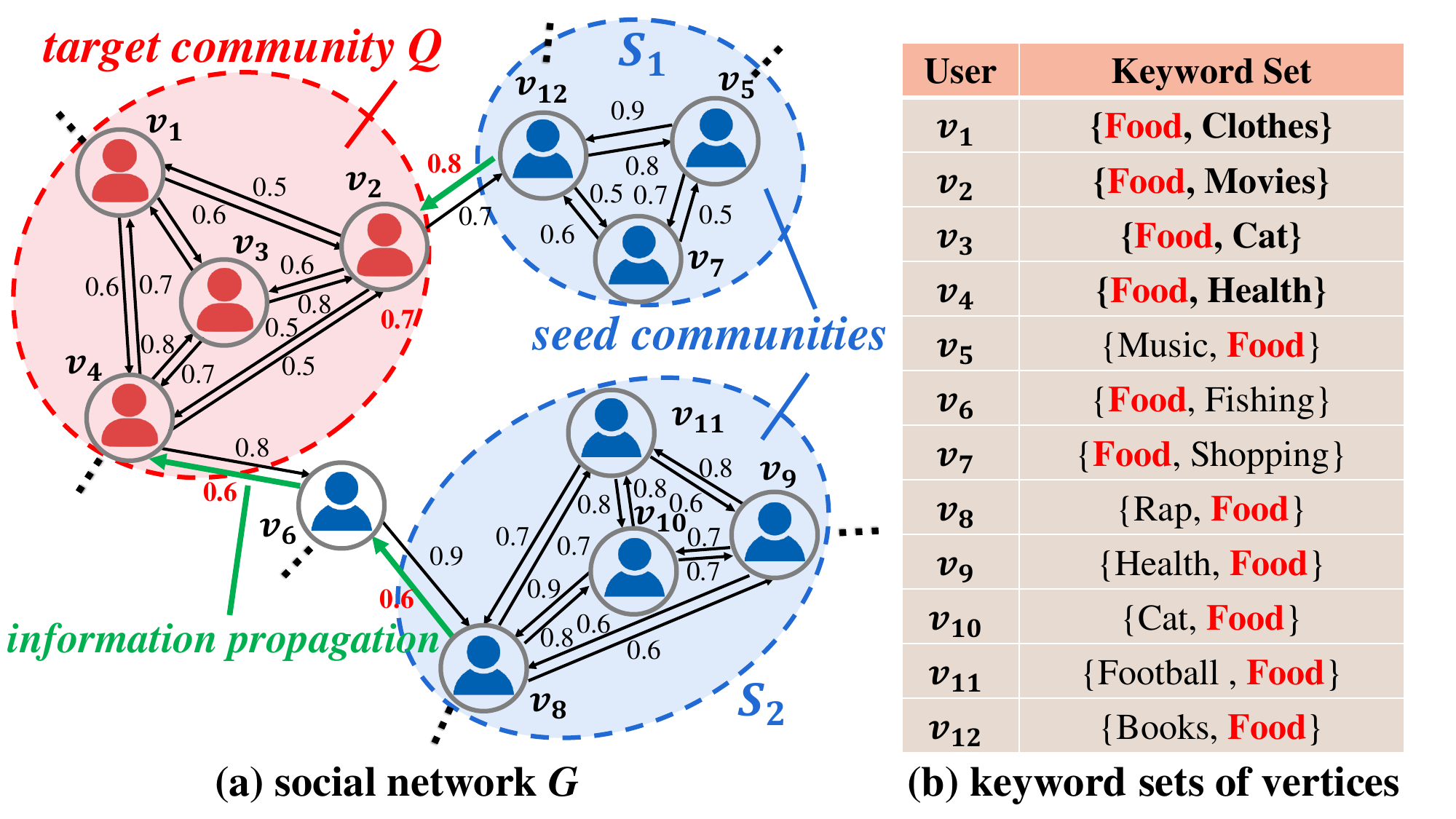}\vspace{-2ex}
    \caption{\label{f1:topk} An example of the Top$M$-RICS problem over social network $G$ ($M=2$, query keyword set $L_q=\{Food\}$). }
    \vspace{-3ex}
\end{figure}

\begin{example}
    \label{example-online}
    \textbf{\textit{(Online Advertising and Marketing Over Social Networks)}} In social networks (e.g., Twitter, Groupon~\cite{wang2017game}), a sales manager wants to ensure the optimal advertisement dissemination of some products (e.g., food) to a targeted group of users through social networks. Figure~\ref{f1:topk}(a) shows an example of social network $G$, where each user vertex $v_i$ ($1\leq i\leq 12$) is associated with a set of keywords (indicating the user's interests, as depicted in Fig.~\ref{f1:topk}(b)). For example, user $v_3$ is interested in delicious food and cute cats. However, direct advertising to the targeted users may sometimes have negative effects, since they might be bored and haunted by repeated promotion advertisements (without even seeing actual merits of the product). In this scenario, the sales manager can specify a group of targeted customers (e.g., $v_1\sim v_4$) for online advertising and marketing (forming a target community $Q$), and issue a Top\textit{M}-RICS query to identify $M$ communities of users (e.g., $S_1$ and $S_2$ in Figure~\ref{f1:topk}(a), where $M=2$) who have the highest impact on the targeted customers in $Q$ (e.g., via tweets/retweets in Twitter). Users in the returned communities $S_1$ and $S_2$ will be given coupons or discounts to promote the products on social networks, and most importantly, indirectly affect the targeted customers' purchase decisions. \qquad $\blacksquare$
\end{example}

\begin{example}
    \textcolor{black}{\textbf{\textit{(Disease Spread Prevention via Contact Networks)}} In real application of infectious disease prevention, there exists some community of vulnerable people (e.g., senior/minor people) who are either reluctant or unable to take preventive actions such as vaccines, due to religion, age, and/or health reasons. The health department may want to identify a few groups of people (e.g., relatives, or colleagues) through contact networks \cite{firestone2011importance} who are most likely to spread infectious diseases to such a vulnerable community, and persuade them to use preventive means (e.g., COVID-19 vaccine). In this case, the health department can exactly perform a Top\textit{M}-RICS query to obtain $M$ communities of people who have the highest disease spreading possibilities to the targeted vulnerable community $Q$.} \qquad $\blacksquare$
\end{example}

\nop{

\begin{example}
    {\color{blue} \textbf{\textit{(Political Campaign)}} In a political campaign, candidate Mike aims to maximize support from specific voters through social networks. Specifically, Mike hopes to identify the most influential communities for targeted advertising and information dissemination to gain more voter support. Most users share common interests such as environmental protection and food safety and are particularly focused on a specific user group. More importantly, such user groups, as seed communities, can maximize the reach of information dissemination to attract the target audience. In this scenario, Mike needs to use the RICS algorithm to identify the most influential seed communities that can impact the target community on the social network.} \qquad $\blacksquare$
\end{example}

}

The Top\textit{M}-RICS problem has many other real applications, such as finding $M$ groups of researchers with the highest influences on another target research community in bibliographical networks.

Inspired by the examples above, in this paper, we consider the Top\textit{M}-RICS problem, which obtains \textit{M} communities (called \textit{seed communities}) that contain query keywords (e.g., food and clothing) and have the highest influences on a target user group. The resulting Top\textit{M}-RICS communities contain highly influential users to whom we can promote products for online advertising/marketing, or suggest taking vaccines for protecting vulnerable people in contact networks.

Note that, efficient and effective answering of Top\textit{M}-RICS queries is quite challenging. A straightforward method is to enumerate all possible communities (subgraphs), compute the influence score of each community with respect to the target user group, and return \textit{M} communities with the highest influence scores. However, this approach is not feasible in practice, due to the large number of candidate communities.

To the best of our knowledge, previous works have not considered the influences on a target user group. Therefore, previous techniques cannot work directly on our Top\textit{M}-RICS problem. To address the challenges of our Top\textit{M}-RICS problem, we propose a two-stage Top\textit{M}-RICS query processing framework in this paper, including offline pre-computation and online Top\textit{M}-RICS querying. In particular, we propose effective pruning strategies (w.r.t., query keyword, boundary support, and influence score) to safely filter out invalid candidate seed communities and reduce the Top\textit{M}-RICS problem search space. Moreover, we design an effective indexing mechanism to integrate our pruning methods seamlessly and develop an efficient algorithm for Top\textit{M}-RICS query processing.

Furthermore, we consider a variant of Top\textit{M}-RICS, named \textit{Top-$M$ Relaxed Reverse Influential Community Search} (Top\textit{M}-R$^2$ICS), which retrieves $M$ subgraphs that have the highest influences on the target user group. Different from Top\textit{M}-RICS, the returned $M$ Top\textit{M}-R$^2$ICS subgraphs have the relaxed structural constraints and the highest influence scores. In Example~\ref{example-online}, since not all members of communities may highly influence the targeted user group $Q$, we thus relax the structural constraints (e.g., $k$-truss \cite{cohen2008trusses} and radius constraint) by returning $M$ subgraph answers that are less structurally cohesive. As an example in Figure~\ref{f1:topk}(a), for community $S_2$, if we replace $v_8$ with $v_6$ in $S_2$, then we can obtain a new (disconnected) subgraph (i.e., potential Top$M$-R$^2$ICS answer), yet with higher impact on target users in $Q$.


In this paper, we make the following major contributions.
\begin{enumerate}
    \item We formally define the \textit{top-M reverse influential community search} (Top\textit{M}-RICS) problem on social networks in Section \ref{def:Problem definition}.
    \item We design an efficient query processing framework for answering Top\textit{M}-RICS queries in Section \ref{RICS framework}.
    \item We propose effective pruning strategies to reduce the Top\textit{M}-RICS problem search space in Section \ref{sec-pruning}.
    \item We devise offline pre-computation and indexing mechanisms in Section \ref{sec-offline} to facilitate pruning and online Top\textit{M}-RICS algorithms in Section \ref{sec-online}.
    \item We formulate a variant, \textit{top-M relaxed reverse influential community search} (Top\textit{M}-R$^2$ICS), and 
    develop an efficient online Top\textit{M}-R$^2$ICS processing algorithm to retrieve $M$ subgraph answers with the relaxed structural constraints in Section~\ref{sec: extension}.
    \item We demonstrate through extensive experiments the effectiveness and efficiency of our Top\textit{M}-RICS/Top\textit{M}-R$^2$ICS query processing algorithms over real/synthetic graphs in Section \ref{Experiment}.
\end{enumerate}

\section{Problem Definition}
\label{def:Problem definition}

This section first gives the data model for social networks with the information propagation in Section \ref{subsec:social_network}, then provides the definitions of target and seed communities in social networks in Section \ref{subsec:community}, and finally formulate a novel problem of \textit{Top-M Reverse Influential Community Search} (Top\textit{M}-RICS) over social networks in Section \ref{subsec:RICS}.


\subsection{Social Networks}
\label{subsec:social_network}

This subsection models social networks by a graph below.

\begin{definition}
(\textbf{Social Network}, $G$) A social network $G$ is a connected graph in the form of a triple $(V(G), E(G),\Phi(G))$, where $V(G)$ and $E(G)$ represent the sets of vertices (users) and edges (relationships between users) in the graph $G$, respectively, and $\Phi(G)$ is a mapping function: $V(G) \times V(G) \rightarrow E(G)$. Each vertex $v_i \in V(G)$ has a keyword set $v_i.L$, and each edge $e_{u,v} \in E(G)$ is associated with an activation probability $P_{u,v}$.
\label{def:social}
\end{definition}


In a social-network graph $G$ (given by Definition \ref{def:social}), each user vertex $v_i$ contains topic keywords (e.g., user-interested topics like movies and sports) in a set $v_i.L$, and each edge $e_{u,v}$ is associated with an activation probability, $P_{u,v}$, which indicates the influence from user $u$ to user $v$ through edge $e_{u,v}$. Here, the activation probability, $P_{u,v}$, can be obtained based on node attributes (e.g., interests, trustworthiness, locations) \cite{min2020topic}, network topology (e.g., node degree, connectivity) \cite{ali2022leveraging,chen2009efficient}, or machine learning techniques \cite{fang2014topic}.








%

\noindent {\bf Information Propagation Model:} In social networks $G$, we consider an information propagation model defined below.

\begin{definition}
 (\textbf{Information Propagation Model}) Given an acyclic path $Path_{u,v}= u_1 \to u_2 \to \cdots \to u_m$ between vertices $u$ ($=u_1$) and $v$ ($=u_m$) in the social network $G$, we define the influence propagation probability, $\mathit{Pr}(Path_{u,v})$, from $u$ to $v$ as:
\begin{equation}
    \mathit{Pr}(Path_{u,v})=\prod_{i=1}^{m-1} P_{u_i, u_{i+1}},
    \label{equ:PRuv}
\end{equation}
where $P_{u_i, u_{i+1}}$ is the activation probability from vertex $u_i$ to vertex $u_{i+1}$.
\end{definition}


Following the \textit{maximum influence path} (MIP) model \cite{chen2010scalable}, an MIP, $\mathit{MIP}_{u,v}$, is a path from $u$ to $v$ with  the highest influence propagation probability (among all paths $Path_{u,v}$), which is:
\begin{equation}
    \mathit{MIP}_{u,v} = \mathop{\arg\max}\limits_{\forall Path_{u,v}}\mathit{Pr}(Path_{u,v}).
    \label{MIP}
\end{equation}

The \textit{influence score}, $\mathit{inf\_score}_{u,v}$, from vertex $u$ to vertex $v$ in the social network $G$ is given by:
\begin{equation}
    \mathit{inf\_score}_{u,v} = Pr(\mathit{MIP}_{u,v}).
    \label{infSuv}
\end{equation}



\subsection{Community}
\label{subsec:community}

In this subsection, we formally define two terms, \textit{target} and \textit{seed community}, as well as the influence from a seed community to a target community, which will be used for formulating our Top\textit{M}-RICS problem.






\noindent {\bf Target Community:} A \textit{target community} is a group of users whom we would like to influence. For example, in the real application of online advertising/marketing, the target community contains the targeted customers to whom we would like to promote some products; for disease prevention, the target community may contain vulnerable people (e.g., senior/minor people) whom we want to protect from infectious diseases.

Formally, we define the target community as follows.

\begin{definition}
(\textbf{Target Community}) Given a social network $G$, a center vertex $v_q$, a list, $L_q$, of query keywords, and the maximum radius $r$, a target community, $Q$, is a connected subgraph of $G$ (denoted as $Q \subseteq G$), such that:
    \begin{itemize}
        \item $v_q \in V(Q)$;
        \item for any vertex $v_i \in V(Q)$, we have $dist(v_q,v_i) \leq r$, and;
        \item for any vertex $v_i \in V(Q)$, its keyword set $v_i.L$ contains at least one query keyword in $L_q$ (i.e., $v_i.L \cap L_q \neq 	\emptyset$),
    \end{itemize}
where $dist(x,y)$ is the shortest path distance between $x$ and $y$ in $Q$.
\label{def:target_community}
\end{definition}


\noindent {\bf Seed Community:} In addition to directly influence the target community (e.g., advertising to targeted users in social networks, or protecting vulnerable people in contact networks), we can also find a group of other users in $G$ (for advertising or protecting, resp.) that indirectly and highly influence the target community. Such a group of influential users forms a \textit{seed community}.

\begin{definition}
(\textbf{Seed Community}) Given a social network $G$, a set, $L_q$, of query keywords, a center vertex $v_s$, an integer parameter $k$, the maximum number, $N$, of community users, and the maximum radius $r$, a seed community, $S_l$, is a connected subgraph of $G$ (denoted as $S_l \subseteq G$), such that:
    \begin{itemize}
        \item $v_s \in V(S_l)$;
        \item $|V(S_l)| \leq N$;
        \item $S_l$ is a $k$-truss \cite{cohen2008trusses}; 
        \item for any vertex $v_i \in V(S_l)$, we have $dist(v_s,v_i) \leq r$, and;
        \item for any vertex $v_i \in V(S_l)$, its keyword set $v_i.L$ contains at least one query keyword in $L_q$ (i.e., $v_i.L \cap L_q \neq \emptyset$).
    \end{itemize}
\label{def:seed_community}
\end{definition}

In Definition \ref{def:seed_community}, the seed community $S_l$ follows the $k$-truss structural constraint \cite{cohen2008trusses,huang2017attribute}, that is, two ending vertices of each edge in the community $S_l$ have at least $(k-2)$ common neighbors (in other words, each edge is contained in at least $(k-2)$ triangles). This $k$-truss requirement indicates the dense structure of the seed community. Note that, our Top\textit{M}-RICS problem can also be extended to other structural constraints such as $k$-core~\cite{yang2012defining}, $k$-clique~\cite{Tsourakakis15a}, and so on, by designing the corresponding offline pre-computation and pruning strategies. We would like to leave this interesting topic of considering other structural constraints as our future work.

\noindent {\bf The Calculation of the Community-to-Community Influence:} We next define the community-level influence, $inf\_score_{S_l, Q}$, from a seed community $S_l$ to a target community $Q$ (w.r.t. topic keywords in $L_q$) in social networks $G$.

\begin{definition}
(\textbf{Community-to-Community Influence}) Given a target community $Q$, a seed community $S_l$, the community-to-community influence, $\mathit{inf\_score}_{Q, S_l}$, of seed community $S_l$ on target community $Q$ is defined as:
\begin{equation}
    \mathit{inf\_score}_{S_l,Q}=\sum_{u \in V(S_l)} \sum_{v \in V(Q)} \mathit{inf\_score}_{u,v},
    \label{sigma_{C_i,C_j}}
\end{equation}
where $\mathit{inf\_score}_{u,v}$ is the influence of vertex $u$ on vertex $v$ (as given in Equation~(\ref{infSuv})).
\label{def:c2c_influence}
\end{definition}

Intuitively, in Definition \ref{def:c2c_influence}, the community-to-community influence $\mathit{inf\_score}_{S_l, Q}$ (as given in Equation~(\ref{sigma_{C_i,C_j}})) calculates the summed influence for all user pairs (in other words, collaborative influence from users in seed community $S_l$ to that in target community $Q$).

\subsection{The Problem Definition of Top-M Reverse Influential Community Search Over Social Networks}
\label{subsec:RICS}
In this subsection, we propose a novel problem, named \textit{Top-M Reverse Influential Community Search} (Top\textit{M}-RICS) over social networks, which retrieves top-$M$ seed communities with the highest influences on a given target community in a social network $G$. 

\noindent {\bf The Top\textit{M}-RICS Problem Definition:} Formally, we have the following Top\textit{M}-RICS problem definition.

\begin{definition}
(\textbf{Top-M Reverse Influential Community Search Over Social Networks, TopM-RICS})
\label{def:RICS}
Given a social network $G = (V(G), E(G), \Phi(G))$, a set, $L_q$, of query keywords, an integer parameter $k$, the maximum number, $N$, of community users, and a target community $Q$ (with center vertex $v_q$, radius $r$, and query keywords in $L_q$), the problem of \textit{top-M reverse influential community search} (TopM-RICS) retrieves a list, $\mathbb{L}$, of $M$ seed communities, $S_l$ (for $1 \leq l \leq M$), from the social network $G$, such that:
\begin{itemize}
    \item $S_l$ satisfies the constraints of seed communities (as given in Definition \ref{def:seed_community}), and;
    \item $M$ seed communities $S_l$ have the highest community-to-community influences, $\mathit{inf\_score}_{S_l,Q}$. 
\end{itemize}
\end{definition}

Intuitively, the Top\textit{M}-RICS problem retrieves keyword-aware seed communities $S_l$ that have the highest influences on the target community $Q$. In real applications such as online advertising/marketing, we can issue the Top\textit{M}-RICS query over the social network $G$ and obtain \textit{M} seed communities $S_l$ of users to whom we can give group buying coupons or discounts to (indirectly) influence the targeted customers in the target community $Q$.

Table \ref{tab1} lists the commonly used notations and their descriptions in this paper. 

\begin{table}[t!]
\caption{Symbols and Descriptions}
\vspace{-0.15in}
\footnotesize
\label{tab1}
\begin{center}
\begin{tabular}{|l|p{6cm}|}
\hline
\textbf{Symbol}&{\textbf{Description}} \\
\hline\hline
$G$ & a social network\\
\hline
$V(G)$ & a set of vertices $v_i$\\
\hline
$E(G)$ & a set of edges $e(u,v)$\\
\hline
$\Phi(G)$ & a mapping function $V(G) \times V(G) \rightarrow E(G)$\\
\hline
$S_l$ (or $Q$) & a seed community (or target community) in $G$\\
\hline
$\mathbb{L}$ & a list of top-$M$ seed communities \\
\hline
$L_q$ & a set of query keywords\\
\hline
$v_i.L$ & a set of keywords associated with user $v_i$\\
\hline
$v_i.BV$ & a bit vector with the hashed keywords in $v_i.L$\\
\hline
$Path_{u,v}$ &  an acyclic path from user $u$ to user $v$\\
\hline
$Pr(Path_{u,v})$ & the propagation probability that user $u$ activates user $v$ through an acyclic path $Path_{u,v}$\\
\hline
$inf\_score_{u,v}$& the influence score of vertex $u$ on vertex $v$\\
\hline
$inf\_score_{S_l,Q}$ &  the community-to-community influence of $S_l$ on $Q$\\
\hline
$r$-$hop(v_i,G)$ & a subgraph in $G$ with $v_i$ as the vertex and $r$ as the radius
\\
\hline
$r$ &  the user-specified radius of target and seed communities\\
\hline
$k$ &  the support parameter in $k$-truss for the seed community\\
\hline
$sup(e_{u,v})$ & the support of edge $e_{u,v}$\\
\hline
$\theta$ & the influence threshold \\
\hline
\end{tabular}
\end{center}\vspace{-3ex}
\end{table}

\section{The Top\textit{M}-RICS Framework}
\label{RICS framework}
Algorithm \ref{alg_RICS} presents our framework for efficiently processing the Top\textit{M}-RICS query, which consists of two phases, that is, \textit{offline pre-computation} and \textit{online Top\textit{M}-RICS computation} phases.

During the offline pre-computation phase, we pre-calculate some data from social networks (for effective pruning) and construct an index over the pre-computed data, which can be used for subsequent online Top\textit{M}-RICS processing. Specifically, for each vertex $v_i$ in the social network $G$, we first hash its set, $v_i.L$, of keywords into a bit vector $v_i.BV$ (lines 1-2). We also pre-calculate a distance vector, $v_i.Dist$, which stores the shortest path distances from vertex $v_i$ to pivots $piv \in S_{piv}$, where $S_{piv}$ is a set of $d$ carefully selected pivot vertices (line 3). Next, we pre-compute the support bounds, boundary influence upper bound, and influence set for $r$-hop subgraphs (centered at vertex $v_i$ and with radii $r$ ranging from 1 to $r_{max}$), in order to facilitate the pruning (lines 4-7). Afterward, we construct a tree index $\mathcal{I}$ on the pre-computed data (line 8).

During the online Top\textit{M}-RICS computation phase, for each user-specified Top\textit{M}-RICS query, we traverse the index $\mathcal{I}$ and apply our proposed pruning strategies (w.r.t. keywords, support, and influence score) to obtain candidate seed communities (lines 9-10). Finally, we calculate the influence scores between candidate seed communities and target community $Q$ to obtain the $M$ seed communities with the highest influential scores (line 11).

\begin{algorithm}[!ht]\small
\caption{{\bf Top\textit{M}-RICS Processing Framework}\small}
\label{alg_RICS}\footnotesize
\KwIn{
    \romannumeral1) a social network $G$,
    \romannumeral2) a set, $L_{q}$, of query keywords,
    \romannumeral3) the maximum radius, $r$, of each community,
    \romannumeral4) an integer parameter, $k$, of the $k$-truss,
    \romannumeral5) the maximum user number, $N$, for each seed community,
    \romannumeral6) the query center vertex $v_q$,
    \romannumeral7) a set, $S_{piv}$,  of pivots, and
    \romannumeral8) an integer parameter, $M$
}

\KwOut{
    a list, $\mathbb{L}$, of top-\textit{M} seed communitites
}

\tcp{\bf offline pre-computation phase}
\For{each $v_i \in V(G)$}{
    hash keywords in $v_i.L$ into a bit vector $v_i.BV$

    compute a vector, $v_i.Dist$, of distances from $v_i$ to all pivots $piv \in S_{piv}$
    
    \For{$r=1$ to $r_{max}$}{
        extract $r$-hop subgraph $r$-$hop(v_i,G)$
        
        compute the upper bound of support $ub\_sup(.)$ in $r$-$hop(v_i,G)$

        compute the upper bound of boundary influence $ub\_bound\_inf_r(.)$ in $r$-$hop(v_i,G)$
    }
}
build a tree index $\mathcal{I}$ over graph $G$ with pre-computed data as aggregates

\tcp{\bf online Top\textit{M}-RICS computation phase}
\For{each Top\textit{M}-RICS query}{
    
    traverse the tree index $\mathcal{I}$ by applying keyword, support, and influence score pruning strategies to retrieve candidate seed communities

    calculate the influence scores of candidate seed communities and return top-$M$ communities with the highest influential scores
}
\end{algorithm}

\section{Pruning strategies}
\label{sec-pruning}

In this section, we present effective pruning strategies that reduce the problem search space during the online Top\textit{M}-RICS computation phase (lines 9-11 of Algorithm \ref{alg_RICS}).

\subsection{Keyword Pruning}
\label{keyword_pruning_4_1}
According to Definitions \ref{def:target_community} and \ref{def:seed_community}, each vertex in the target/seed community $Q$ or $S_l$ must contain at least one keyword from the query keyword set $L_q$. Therefore, our keyword pruning method can filter out those candidate subgraphs that do not meet this criterion.

\begin{lemma}
    \label{lemma:keyword_pruning}
    {\bf (Keyword Pruning)} Given a set, $L_q$, of query keywords and a candidate subgraph (community) $S_l$, any vertex  $v_i \in V(S_l)$ can be safely pruned from $S_l$, if it holds that: $v_i.L \cap L_q = \emptyset$, where $v_i.L$ is the keyword set associated with vertex $v_i$.
\end{lemma}
\begin{proof}
    If $v_i.L \cap L_q = \emptyset$ holds for any user vertex $v_i$ in a candidate community $S_l$, it indicates that user $v_i$ is not interested in any keyword in the query keyword set $L_q$. Thus, user vertex $v_i$ does not satisfy the keyword constraint in  Definition~\ref{def:seed_community}, and vertex $v_i$ can be safely pruned from $S_l$, which completes the proof. \qquad $\square$
\end{proof}

\subsection{Support Pruning}

From Definition \ref{def:seed_community}, the seed community $S_l$ needs to be a \textit{k}-truss \cite{cohen2008trusses}. Denote the support, $sup(e_{u,v})$ of an edge $e_{u,v}$ as the number of triangles containing $e_{u,v}$. Each edge $e_{u,v}$ in the seed community $S_l$, is required to have its support $sup(e_{u,v})$ greater than or equal to $(k-2)$. If we can obtain an upper bound, $ub\_sup(e_{u,v})$, of the support for each edge in the candidate seed community $S_l$, then we can employ the following lemma to eliminate candidate seed communities with low support.

\begin{lemma}
    \label{lemma:support_pruning}
    {\bf (Support Pruning)} Given a candidate seed community $S_l$ and a positive integer $k$ $(>2)$, an edge $e_{u,v}$ in $S_l$ can be discarded safely from $S_l$, if it holds that $ub\_sup(e_{u,v}) < k-2$, where $ub\_sup(e_{u,v})$ is an upper bound of the edge support $sup(e_{u,v})$.
\end{lemma}
\begin{proof}
    In the definition of the $k$-truss \cite{cohen2008trusses}, the support value, $sup(e_{u,v})$, of the edge $e_{u,v}$ is determined by the number of triangles that contain edge $e_{u,v}$. In a $k$-truss, each edge must be reinforced by at least $(k-2)$ such triangle structures. 
    Since we have the conditions that $ub\_sup(e_{u,v}) < k-2$ (lemma assumption) and $sup(e_{u,v}) \leq ub\_sup(e_{u,v})$ (support upper bound property), by the inequality transition, we have $sup(e_{u,v}) < k-2$. Therefore, based on Definition \ref{def:seed_community}, the $k$-truss seed community $S_l$ cannot include the edge $e_{u,v}$ due to its low support (i.e., $<k-2$). We thus can safely rule out edge $e_{u,v}$ from $S_l$, which completes the proof. \qquad $\square$
\end{proof}

\subsection{Influence Score Pruning}

In this subsection, we provide an effective pruning method to filter out candidate seed communities with low influence scores.

Since the exact calculation of the influence score between two communities (given by Equation~(\ref{sigma_{C_i,C_j}})) is very time-consuming, we can take the maximum influence from candidate seed communities that we have seen as an influence score upper bound (denoted as an influence threshold $\theta$, which is the $M$-th highest influence score for seed communities we have seen so far). This way, we can apply the influence score pruning in the lemma below to eliminate those seed communities with low influences.


\begin{lemma}
    \label{lemma:influence_pruning}
    {\bf (Influence Score Pruning)} Let an influence threshold $\theta$ be the $M$-th highest influence from candidate seed communities we have obtained so far on the target community $Q$. Any candidate seed community $S_l$ can be safely pruned, if it holds that $ub\_inf\_score_{S_l,Q} < \theta$, where $ub\_inf\_score_{S_l,Q}$ is an upper bound of the influence score $inf\_score_{S_l,Q}$ (given by Equation~(\ref{sigma_{C_i,C_j}})).
\end{lemma}
\begin{proof}
Since $ub\_inf\_score_{S_l,Q}$ is an upper bound of the influence score $inf\_score_{S_l,Q}$, we have $inf\_score_{S_l,Q}\leq ub\_inf\_score_{S_l,Q}$. Due to the lemma assumption that $ub\_inf\_score_{S_l,Q} < \theta$, by the inequality transition, it holds that $inf\_score_{S_l,Q}< \theta$, which indicates that the candidate community $S_l$ has a lower influence on $Q$ than that of at least $M$ communities we have obtained so far (i.e., the influence threshold $\theta$). In other words, $S_l$ cannot be one of our Top\textit{M}-RICS answers. Therefore, we can safely prune candidate seed community $S_l$ if $ub\_inf\_score_{S_l,Q} < \theta$ holds, which completes the proof.   \quad $\square$
\end{proof}

\section{Offline pre-computation}
\label{sec-offline}
In this section, we discuss how to offline pre-compute data over social networks, and construct a tree index $\mathcal{I}$ on pre-computed data (lines 1-8 of Algorithm \ref{alg_RICS}).

\begin{algorithm}[!ht]\small
\caption{{\bf Offline Pre-Computation}\small}
\label{alg_pffline}\footnotesize
\KwIn{
    \romannumeral1) a social network $G$;
    \romannumeral2) the maximum radius, $r_{max}$, of each community, and;
    \romannumeral3) a set, $S_{piv}$, of $d$ pivots
}
\KwOut{
    pre-computed auxiliary data $v_i.Aux$ for each vertex $v_i$
}

\For{each $v_i \in V(G)$}{
    \tcp{the keyword bit vector}

    hash keywords in $v_i.L$ into a bit vector $v_i.BV_0$

    $v_i.Aux = \{v_i . BV_0\}$

    \tcp{the distance vector to pivots}

    compute a vector, $v_i.Dist$, of distances from vertex $v_i$ to $d$ pivots in $S_{piv}$


    add $v_i.Dist$ to $v_i.Aux$

    \tcp{edge support upper bounds}

    \For{each $e_{u,v} \in E(r_{max}$-$hop(v_i,G))$}{
            compute on edge support upper bound $ub\_sup(e_{u,v})$
        }
}

\For{each $v_i \in V(G)$}{
    
    \For{$r=1$ to $r_{max}$}{

        $v_i.BV_r=\bigvee_{\forall{v_l \in {r\text{-}hop(v_i,G)}}}v_l.BV$
        
        $v_i.ub\_sup_r = \max_{\forall e_{u,v} \in E(r\text{-}hop(v_i,G))}ub\_sup(e_{u,v})$

        $v_i.ub\_bound\_inf_r = max\{{\sf{collapse\_calculate}}(r\text{-}hop(v_i,G))\}$
        
        add $v_i.BV_r$ , $v_i.ub\_sup_r$ and $v_i.ub\_bound\_inf_r$ to $v_i.Aux$
    }
}
\Return $v_i.Aux$
\end{algorithm}


\subsection{Offline Pre-Computed Data}
\label{sec:offline_precomputed_data}
In order to facilitate online Top\textit{M}-RICS computation, we first conduct offline pre-computations on the social network $G$ in Algorithm \ref{alg_pffline}, which can obtain aggregated information about candidate seed communities (later used for pruning strategies to reduce the online search cost). Specifically, for each vertex $v_i$, we hash a set, $v_i.L$, of its keywords into a bit vector $v_i.BV_0$ of size $B$, and initialize a pre-computed set, $v_i.Aux$, of auxiliary data with $v_i.BV_0$ (lines 1-3). Then, we compute the distances from $v_i$ to $d$ pivots in $S_{piv}$, forming a distance vector $v_i.Dist$ of size $d$, and add $v_i.Dist$ to $v_i.Aux$ (lines 4-5). Next, we compute a support upper bound, $ub\_sup(e_{u,v})$, for each edge $e_{u,v}$ in a subgraph, $r_{max}\text{-}hop(v_i,G)$, centered at vertex $v_i$ and with radius $r_{max}$ (lines 6-7). Then, for each vertex $v_i$ and possible radius $r \in [1, r_{max}]$, we pre-compute a keyword bit vector (lines 8-10), an edge support upper bound (line 11), and an upper bound, $v_i.ub\_bound\_inf_r$, of boundary influence scores (line 12) for $r\text{-}hop(v_i, G)$ subgraph. Finally, we add these pre-computed aggregated information to $v_i.Aux$ in the following format: $\{v_i.BV_0, v_i.Dist, v_i.BV_r, v_i.ub\_sup_r,v_i.ub\_bound\_inf_r\}$ (line 13). 

To summarize, $v_i.Aux$ contains the following information:
\begin{itemize}
    \item {\bf a bit vector, \bm{$v_i.BV_0$}, of size \bm{$B$}}, which is obtained by using a hashing function $f(l)$ to hash each keyword $l \in v_i.L$ to an integer between $[0, B-1]$ and set the $f(l)$-th bit position to 1 (i.e., $v_i.BV_0[f(l)]=1$);
    \item {\bf a distance vector, \bm{$v_i.Dist$}, of size \bm{$d$}}, which is obtained by computing the shortest path distances, $dist(v_i, piv)$, from $v_i$ to $d$ pivots $piv_j \in S_{piv}$; (i.e., $v_i.Dist[j]=dist(v_i, piv_j)$ for $0\leq j < d$); 
    \item {\bf a bit vector, \bm{$v_i.BV_r$} (for \bm{$1\leq r\leq r_{max}$})}, which is obtained by hashing each keyword in keyword set $v_l.L$ of a vertex $v_l$ in the subgraph $r\text{-}hop(v_i,G)$ into a position in the bit vector (i.e., $v_i.BV_r=\bigvee_{\forall{v_l \in r\text{-}hop(v_i,G)}}v_l.BV$);
    \item {\bf a support upper bound, \bm{$v_i.ub\_sup_r$}}, which is obtained by taking the maximum of all support bounds $ub\_sup(e_{u,v})$ for edges $e_{u,v}$ in the subgraph $r\text{-}hop(v_i,G)$ (i.e., $v_i.ub\_sup_r = \max_{\forall e_{u,v} \in E(r\text{-}hop(v_i,G)))}ub\_sup(e_{u,v})$), and;
    \item {\bf an upper bound,  \bm{$v_i.ub\_bound\_inf_r$}, of boundary influence scores}, which is obtained by computing the virtual collapse of a subgraph $r\text{-}hop(v_i,G)$ discussed below (i.e., $v_i.ub\_bound\_inf_r = max\{{\sf{collapse\_calculate}}(r\text{-}hop(v_i,G))\}$), {\color{black} where function {\sf collapse\_calculate}($\cdot$) returns a set, $v_i.BIS$, of influence scores through boundary vertices}.
\end{itemize}

\begin{figure}
    \centering
    \vspace{-3ex}
    \includegraphics[width=0.45\textwidth]{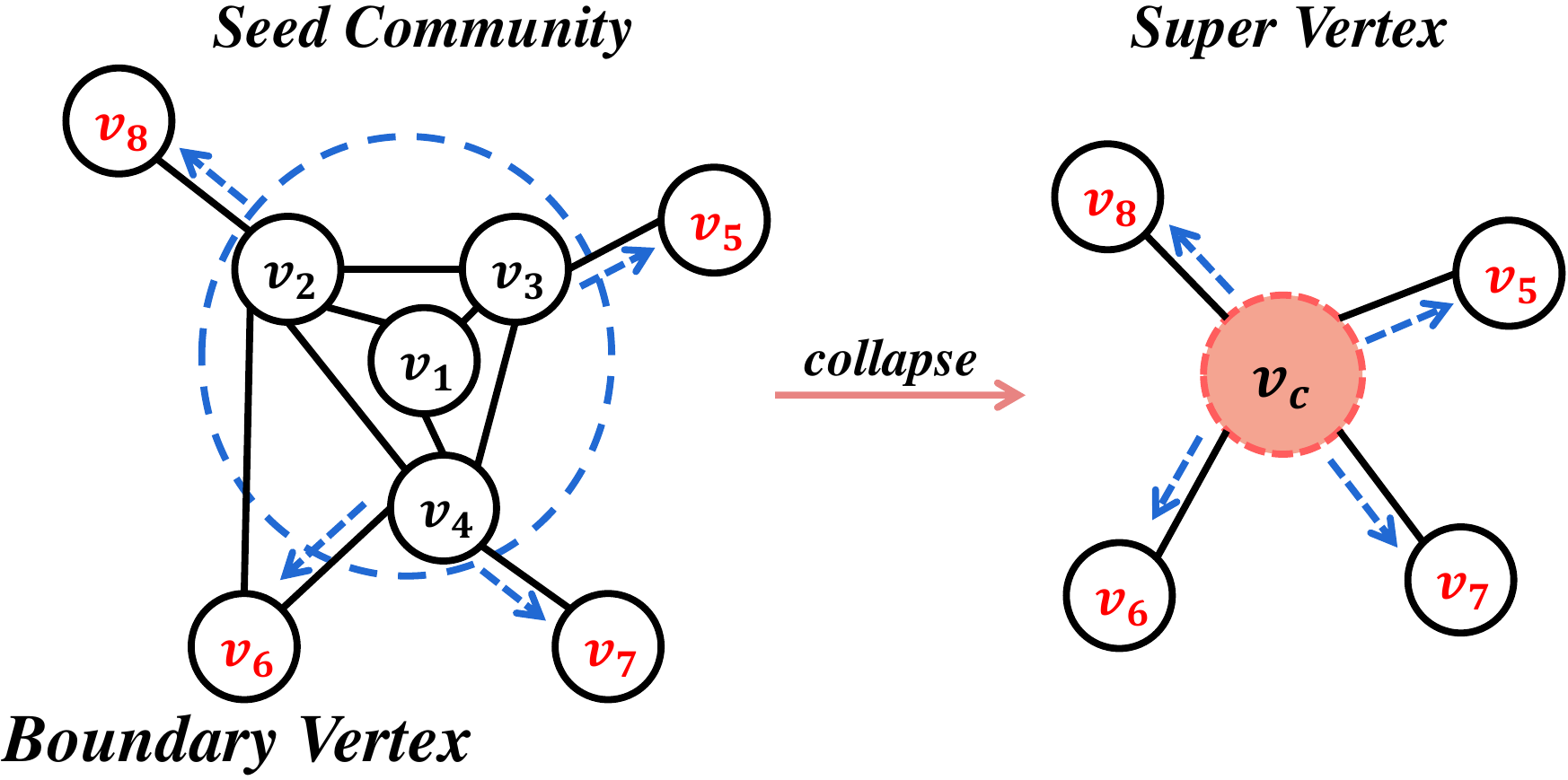}
    \caption{{\color{black}An example of seed community virtual collapse operation. The blue arrows represent the information propagation. $v_c$ represents a virtual super vertex after the community has collapsed. The red vertices represent the boundary vertices.}} 
    \label{fig: collapse}
    \vspace{-3ex}
\end{figure}

{\color{black} \noindent{\bf Discussions on How to Implement $\sf collapse\_calculate(\cdot)$:} Collapse calculations are divided into \textit{target collapse} and \textit{seed collapse}. The difference between the two collapses is that the information propagation is in different directions. As shown in Figure~\ref{fig: collapse}, a seed community consisting of $v_1$, $v_2$, $v_3$, and $v_4$ sends influence to the $1$-hop boundary vertices (i.e., $v_5$, $v_6$, $v_7$, and $v_8$). According to Equation~(\ref{infSuv}) and (\ref{sigma_{C_i,C_j}}), we aggregate the influence of seed communities towards their external boundary vertices. For a virtual collapsed vertex $v_c$ of a seed community $S_l$, through any $1$-hop subgraph boundary vertex $v_i$, we have $\mathit{inf\_score}_{v_c, v_i}=\sum_{v \in V(S_l)} \mathit{inf\_score}_{v, v_i}$. Finally, we store the set of boundary influence scores $v_i.BIS$ in the center vertex $v_i$ of the community.}



{ \color{black}
\noindent{\bf Complexity Analysis:} As shown in Algorithm \ref{alg_pffline}, for each vertex $v_i\in V(G)$ in the first loop, the time complexity of computing a keyword bit vector $v_i.BV$ is given by $O(|L|)$ (lines 2-3). And the time complexity of computing a distance vector $v_i.Dist$ is given by $O((|V(G)|+|E(G)|) \cdot log|V(G)|)$ (lines 4-5). Let $avg\_deg$ denote the average number of vertex degrees. Since there are $avg\_deg^{r_{max}}$ edges in $r_{max}\text{-}hop(v_i, G)$ and the cost of the support upper bound computation is a constant (counting the common neighbors), so the time cost of obtaining all edge support upper bounds is $O(avg\_deg^{r_{max}})$ (lines 6-7).
Thus, the complexity of the first loop (lines 1-7) is given by $O(|V(G)|\cdot (|L|+(|V(G)|+|E(G)|) \cdot log|V(G)|+avg\_deg^{r_{max}}))$.

In the second loop, for each $v_i\in V(G)$, there are $avg\_deg^{r-1}$ vertices in the $r\text{-}hop(v_i, G)$ w.r.t. $r$. Then, for each $r \in [1, r_{max}]$, the time  complexity of computing $v_i.BV_r$ and $v_i.ub\_sup_r$ is given by $O(B\cdot avg\_deg^{r-1})$ and $O(avg\_deg^{r-1})$, respectively (lines 10-11). As described in Section \ref{sec:offline_precomputed_data}, the time complexity of $\sf collapse\_calculate(\cdot)$ is $O((avg\_deg^{r}+avg\_deg^{r-1})\cdot log(avg\_deg^{r-1}))$, and so, the time complexity of $v_i.ub\_bound\_inf_r$ is given by $O(n \cdot (avg\_deg^{r} + avg\_deg^{r-1})\cdot log(avg\_deg^{r-1}))$.
Therefore, the time cost of the second loop (lines 8-13) is $O(|V(G)|\cdot r_{max}\cdot (B\cdot avg\_deg^{r-1}+avg\_deg^{r-1}+n \cdot (avg\_deg^{r} + avg\_deg^{r-1})\cdot log(avg\_deg^{r-1})))$.

In summary, the total time complexity of the total offline pre-computation is given by $O(|V(G)|\cdot (|L|+(|V(G)|+|E(G)|) \cdot log|V(G)|+avg\_deg^{r_{max}}+r_{max}\cdot (B\cdot avg\_deg^{r-1}+avg\_deg^{r-1}+n \cdot (avg\_deg^{r} + avg\_deg^{r-1})\cdot log(avg\_deg^{r-1}))))$.
}

\subsection{Indexing Mechanism} In this subsection, we show the details of offline construction of a tree index $\mathcal{I}$ on a social network $G$ to support online Top\textit{M}-RICS query processing.

\noindent{\bf The Data Structure of Index $\mathcal{I}$:} We will build a tree index $\mathcal{I}$ on the social network $G$, where each index node, $\mathcal N$ includes multiple entries $\mathcal N_i$, each corresponding to a subgraph of $G$.
Specifically, the tree index $\mathcal{I}$ contains two types of nodes, leaf and non-leaf nodes.

\noindent{\it Leaf Nodes:}
{\color{black}
Each leaf node $\mathcal N$ contains multiple vertices $v_i$ in the corresponding subgraph. The community subgraph centered at $v_i$ is denoted by $r\text{-}hop(v_i, G)$. Moreover, each vertex $v_i$ is associated with the following pre-computed data in $v_i.Aux$ (some of them are w.r.t. each possible radius $r\in  [1,r_{max}]$):
\begin{itemize}
    \item a keyword bit vector $v_i.BV_r$;
    \item a distance vector $v_i.Dist$;
    \item a support upper bound $v_i.ub\_sup_r$, and;
    \item a boundary influence upper bound $v_i.ub\_bound\_inf_r$.
\end{itemize}
}


\textit{Non-Leaf Nodes:} Each non-leaf node $\mathcal N$ has multiple index entries, $\mathcal N_i$, each of which is associated with the following aggregates (w.r.t. each possible radius $r \in  [1,r_{max}]$):
\begin{itemize}
    \item a pointer to a child node $\mathcal N_i. ptr$;
    \item an aggregated keyword bit vector $\mathcal N_i.BV_r = \bigvee_{\forall v_l \in \mathcal N_i} v_l.BV_r$;
    \item the distance lower bound vector $\mathcal N_i.lb\_Dist$ (i.e., $N_i.lb\_Dist[j] = \min_{\forall v_l \in \mathcal N_i} v_l.Dist[j]$, for $1\leq j\leq d$);
    \item the distance upper bound vector $\mathcal N_i.ub\_Dist$ (i.e., $N_i.ub\_Dist[j] = \max_{\forall v_l \in \mathcal N_i} v_l.Dist[j]$, for $1\leq j\leq d$);
    \item the maximum support upper bound 
    $\mathcal N_i.ub\_sup_r=\\
    \max_{\forall v_l \in \mathcal N_i} v_l.ub\_sup_r$, and;
    \item the maximum boundary influence upper bound $\mathcal N_i.ub\_bound\_inf_r = \max_{\forall v_l \in \mathcal N_i} v_l.ub\_bound\_inf_r$.
\end{itemize}

{\color{black} \noindent{\bf Index Construction:} To construct the tree index $\mathcal{I}$, we will utilize cost models to first partition the graph into (disjoint) subgraphs of similar sizes to form initial leaf nodes, and then recursively group subgraphs (or nodes) into non-leaf nodes on a higher level, until one final root of the tree is obtained.


}

\noindent{\bf Cost Model for the Graph Partitioning:} {\color{black} Specifically, we use METIS \cite{karypis1998fast} for graph partitioning, guided by our proposed cost model.} Our goal of designing a cost model for the graph partitioning is to reduce the number of cases that candidate communities are across subgraph partitions (or leaf nodes), and in turn achieve low query cost.

Assume that a graph partitioning strategy, $\mathcal{P}$, divides the graph into $m$ subgraph partitions $P_1$, $P_2$, ..., and $P_m$. 
We can obtain the number, $Cross\_Par\_Size(\mathcal{P})$, of \textit{cross-partition vertices for candidate communities} as follows.
\begin{eqnarray}
&&Cross\_Par\_Size(\mathcal{P}) \\
&=& \sum_{j=1}^m\sum_{\forall v_i\in P_j}|V(r_{max}\text{-}hop(v_i,G)-P_j)|\notag
    \label{equ: over partitioning number}
\end{eqnarray}



Since we would like to have the subgraph partitions of similar sizes, we also incorporate the maximum size difference of the resulting partitions in $\mathcal{P}$, and have the following target cost model, $CM(\mathcal{P})$.
\begin{equation}
    CM(\mathcal{P}) = \mathop{\arg\min_\mathcal{P}}(Cross\_Par\_Size(\mathcal{P}) + (|\mathcal{P}_{max}|-|\mathcal{P}_{min}|)),
    \label{equ:CM}
\end{equation}
where $|\mathcal{P}_{max}|$ and $|\mathcal{P}_{min}|$ represent the numbers of users in the largest and smallest partitions in $\mathcal{P}$, respectively. 

Intuitively, we would like to obtain a graph partitioning strategy $\mathcal{P}$ that minimizes our cost model $CM(\mathcal{P})$ (i.e., with low cross-partition search costs and of similar partition sizes, as given in Equation~(\ref{equ:CM})). 

\noindent{\bf Cost-Model-Guided Graph Partitioning for Obtaining Index Nodes:} {\color{black} In Algorithm \ref{alg: index partition selection}, we illustrate how to obtain a set, $\mathbb{S}_{p}$, of $m$ graph partitions for creating index nodes, in light of our proposed cost model above. First, we randomly select $m$ initial vertex pivots and form an initial set, $\mathbb{S}_{piv}$ (line 1). Then, we use $\mathbb{S}_{piv}$ to perform the graph clustering and obtain $m$ partitions in $\mathbb{S}_{p}$ (line 2). We invoke $\sf{calculate\_cost} (\mathbb{S}_{p})$ in Algorithm \ref{alg: calculate_cost} to calculate the cost, $local\_cost$, of the partitioning $\mathbb{S}_{p}$ (i.e., via $CM(\mathbb{S}_{p})$ in Equation~(\ref{equ:CM}) of our cost model; line 3). 

\begin{algorithm}[!ht]\small
\caption{{\bf Cost-Model-Guided Graph Partitioning for Index Nodes}\small}
\label{alg: index partition selection}\footnotesize
\KwIn{
    \romannumeral1) a social network $G$,
    \romannumeral2) the number $m$ of center vertices for partitioning, and
    \romannumeral3) the maximum number of iterations $iter_{max}$
}
\KwOut{
    a set, $\mathbb{S}_{p}$, of $m$ graph partitions for creating index nodes
}

    randomly select $m$ initial vertex pivots and form $\mathbb{S}_{piv}$

    use $\mathbb{S}_{piv}$ clustering to form $m$ partitions $\mathbb{S}_{p}$

    calculate the cost of the partitioning: $local\_cost = \sf{calculate\_cost} (\mathbb{S}_{p})$

    \For{$iter = 1$ to $iter\_max$}{
    
        select a random pivot $piv \in \mathbb{S}_{piv}$

        randomly select a new vertex $piv_{new}$ that satisfies the requirements of $\mathbb{S}_{piv}$

        $\mathbb{S}_{piv}^{'} = \mathbb{S}_{piv}-\{piv\}+\{piv_{new}\}$

        use $\mathbb{S}_{piv}^{'}$ clustering to form $m$ partitions $\mathbb{S}_{p}^{'}$

        calculate the cost of new partitions: $cost_{new} = \sf{calculate\_cost} (\mathbb{S}_{p}^{'})$

        \If{$cost_{new} < local\_cost$}{
            $\mathbb{S}_{piv} = \mathbb{S}_{piv}^{'}$

            $\mathbb{S}_{p} = \mathbb{S}_{p}^{'}$
            
            $local\_cost = cost_{new}$
        }
    }

    \Return $\mathbb{S}_{p}$

\end{algorithm}

Next, we perform $iter\_max$ iterations to find the best pivot set $\mathbb{S}_{piv}$ and graph partitioning $\mathbb{S}_{p}$ with low cost $local\_cost$ (lines 4-13). In each iteration, we randomly replace one of vertex pivots, $piv$, in $\mathbb{S}_{piv}$ with a new non-pivot vertex $piv_{new}$, forming a new pivot set, $\mathbb{S}_{piv}'$ (lines 5-7). This way, we can use $\mathbb{S}_{piv}'$ to perform the graph clustering and obtain $m$ new partitions in $\mathbb{S}_{p}'$, so that we invoke the function $\sf{calculate\_cost} (\mathbb{S}_{p}')$ to calculate a new cost, $cost_{new}$, of partitioning $\mathbb{S}_{p}'$ (lines 8-9). Correspondingly, if $cost_{new}$ is less than $local\_cost$, 
we accept the new partitioning strategy by updating $\mathbb{S}_{piv}$, $\mathbb{S}_{p}$, and $local\_cost$ with $\mathbb{S}_{piv}'$, $\mathbb{S}_{p}'$, and $cost_{new}$, respectively (lines 10-13). Finally, we return $m$ subgraph partitions, $\mathbb{S}_{p}$, to create $m$ index nodes, respectively (line 14).}

\begin{algorithm}[!ht]\small
\caption{\bf calculate\_cost($\cdot$) Function}
\label{alg: calculate_cost}\footnotesize
\KwIn{
    a set, $ \mathbb{S}_{p}$, of partitions over social network $G$
}
\KwOut{
    a cost score, $CM(\mathcal{P})$, for the partitioning in $G$
}
$CM(\mathcal{P}) = 0$

\For{each $P \in  \mathbb{S}_p$}{

    \For{each $v_i \in V(P)$}{

        count the number of vertices that cross the partition $P$'s range: $N\_{cross} = |V(r_{max}\text{-}hop(v_i,G)-P)|$
    }
    $Cross\_Par\_Size(P) = \sum_{\forall v_i \in P}N\_{cross}$
    
    add the value of $Cross\_Par\_Size(P)$ to $CM(\mathcal{P})$
}

add $|\mathcal{P}_{max}|-|\mathcal{P}_{min}|$ to $CM(\mathcal{P})$

\Return $CM(\mathcal{P})$
\end{algorithm}

{\color{black}
\noindent{\bf Complexity Analysis:} For the tree index $\mathcal{I}$, let $\gamma$ denote the fanout of each non-leaf node $\mathcal{N}$. In $\mathcal{I}$, since the number of leaf nodes is equal to the number of vertices $|V(G)|$, the depth of tree index $\mathcal{I}$ is $\lceil \log_{\gamma}{|V(G)|} \rceil + 1$. The time complexity of cost-model-guided graph partitioning for index nodes is given by $O((|V(G)| \cdot m+ |V(G)|)\cdot iter\_max)$. On the other hand, the time complexity of recursive tree index construction is $O((\gamma^{dep}-1)/(\gamma-1) \cdot Partitioning)$. Therefore, the time complexity of our tree index construction is given by $O((\gamma^{\lceil \log_{\gamma}{|V(G)|} \rceil + 1}-1)/(\gamma-1) \cdot {|V(G)| \cdot (m+1)\cdot iter\_max})$.
}

\section{Online Top\textit{M}-RICS Computation}
\label{sec-online}

{\color{black}
In this section, we provide our online Top\textit{M}-RICS computation algorithm in Algorithm~\ref{alg: online_RICS}, which traverses our constructed tree index $\mathcal{I}$ and retrieves the Top\textit{M}-RICS community answer that has the highest influence on the target community $Q$, by seamlessly integrating our effective pruning strategies. 

Section~\ref{Index Pruning} presents effective pruning strategies on the node level of the tree index. Section~\ref{RICS algorithm} details our proposed online Top\textit{M}-RICS query processing procedure.
}


\subsection{Index Pruning}
\label{Index Pruning}

In this subsection, we present effective pruning methods on the index level, which are used to prune index nodes containing (a group of) community false alarms.


\noindent{\bf Keyword Pruning for Index Entries:} The idea of our keyword pruning over index entries is as follows. If all the $r$-hop subgraphs under an index entry $\mathcal{N}_i$ do not contain any keywords in the query keyword set $L_q$, then the entire index entry $\mathcal{N}_i$ can be safely filtered out. 

Below, we provide the \textit{index keyword pruning} method that uses the aggregated keyword bit vector $\mathcal{N}_i.BV_r$ stored in $\mathcal{N}_i$.

\begin{lemma}
    \label{Index Keyword Pruning}
    {\bf (Index Keyword Pruning)} Given an index entry $\mathcal{N}_i$ and a bit vector, $L_q.BV$, for the query keyword set $L_q$, the index entry $\mathcal{N}_i$ can be safely pruned, if it holds that $\mathcal{N}_i.BV_r \wedge L_q.BV = \boldsymbol{0}$. 
\end{lemma}
\begin{proof}
    If $\mathcal{N}_i.BV_r \cap L_q.BV = \emptyset$ holds, which means that all communities in $\mathcal{N}_i$ do not contain any of the keywords in $L_q$. According to Definition \ref{def:seed_community}, $\mathcal{N}_i$ cannot be a candidate seed community, so it can be safely pruned. \qquad $\square$
\end{proof}

\noindent{\bf Support Pruning for Index Entries:} Next, we present the \textit{index support pruning} method, which utilizes the maximum upper bound support $\mathcal{N}_i.ub\_sup_r$ of the index entry $\mathcal{N}_i$ and the given support $k$ to rule out the entry with low support.
\begin{lemma}
    \label{Index Support Pruning}
    {\bf (Index Support Pruning)} Given an index entry $\mathcal{N}_i$ and a support parameter $k$, the index entry $\mathcal {N}_i$ can be safely pruned, if it holds that $\mathcal{N}_i.ub\_sup_r < k$, where $\mathcal{N}_i.ub\_sup_r$ is the maximum support upper bound for all $r$-hop subgraphs under $\mathcal{N}_i$.
\end{lemma}
\begin{proof}
    $\mathcal N_i.ub\_sup_r$ is the maximum support upper bound in all $r$-hop subgraphs under index entry $\mathcal N_i$. If $\mathcal N_i.ub\_sup_r < k$ holds, then all support upper bounds of $r$-hop subgraphs under $\mathcal N_i$ are less than $k$. By the inequality transition, all the supports of $r$-hop subgraphs under entry $\mathcal N_i$ are thus less than $k$. Based on Definition \ref{def:seed_community}, all $r$-hop subgraphs under $\mathcal N_i$ cannot be a candidate seed community. Therefore, index entry $\mathcal N_i$ can be safely pruned, which completes the proof of this lemma. \qquad $\square$

\end{proof}

\subsection{The Top\textit{M}-RICS Algorithm}
\label{RICS algorithm}
In this subsection, we illustrate our online Top\textit{M}-RICS processing algorithm by traversing the tree index $\mathcal{I}$ in Algorithm~\ref{alg: online_RICS}.

\begin{algorithm}[!h]\small
\caption{{\bf Online Top\textit{M}-RICS Processing}\small}
\label{alg: online_RICS}\footnotesize
\KwIn{
    \romannumeral1) a social network $G$,
    \romannumeral2) a set, $L_{q}$, of query keywords,
    \romannumeral3) the maximum radius, $r$, of each community,
    \romannumeral4) an integer parameter, $k$, of the truss for each seed community,
    \romannumeral5) an integer parameter, $N$, of the maximum number of users for each seed community,
    \romannumeral6) the query center vertex, $v_q$,
    \romannumeral7) an integer parameter, $M$, and
    \romannumeral8) the index $\mathcal{I}$
}
\KwOut{
    a list, $\mathbb{L}$, of top-\textit{M} seed communitites
}

\tcp{initialization}

    hash all keywords in the query keyword set $L_{q}$ into a query bit vector $L_{q}. BV$

    obtain the target community $Q = r\text{-}hop(v_q, G)$

    $\mathbb{L}=\emptyset, C_{cand}=\emptyset, \theta = 0$

    \tcp{index traversal}

    initialize a minimum heap $\mathcal{H}$ accepting index entries in the form $(\mathcal{N},key)$

    insert all entries $\mathcal{N}$ in the root of index $\mathcal{I}$ into heap $\mathcal{H}$

    \While{$\mathcal{H}$ is not empty}{

        $(\mathcal{N},key)=\mathcal{H}.pop()$

        \If{$\theta \geq max\_inf\_ub(\mathcal{H})$}{
            terminal the loop
        }
    
        \eIf{$\mathcal{N}$ is a leaf node}{
            \For{each vertex $v_i\in \mathcal{N}$}{
                obtain the candidate community $C = \mathcal{N}.r\text{-}hop(v_i, G)$

                \If{$C$ cannot be pruned by Lemma \ref{lemma:keyword_pruning}, \ref{lemma:support_pruning}, or \ref{lemma:influence_pruning}}
                {

                    \eIf{$C$ is closer to $Q$ and larger than some community in $\mathbb{L}$}{
                            
                    \eIf{$|\mathbb{L}| < M$}{

                        compute the influence score $inf\_score_{C, Q}$ $=$ $\sf{calculate\_influence}$$(C, Q)$ (for multiple trimmed $C$, if $|C|>N$)

                        add $C$ to $\mathbb{L}$

                               {\bf if} $|\mathbb{L}| = M$ {\bf then} update $\theta$ with the minimum influence score in $\mathbb{L}$

                        }{
                            \If{$v_i.ub\_inf\_score_r > \theta$}{
                            
                                add $C$ to $C_{cand}$
                                
                            }
                        }
                    }{
                                add $C$ to $C_{cand}$
                            }
                }
            }
        }{
        \tcp{$\mathcal{N}$ is a non-leaf node}
        
            \For{each entry $\mathcal{N}_i \in \mathcal{N}$}{
            
                \If{$\mathcal{N}_i$ cannot be pruned by Lemma \ref{Index Keyword Pruning} or \ref{Index Support Pruning}}{

                    insert $(\mathcal{N}_i,key)$ into heap $\mathcal{H}$ 
                    
                }
            }
        }
    }

    \tcp{refinement of candidate communities}
    
    update $C_{cand}$ by sorting on influence score upper bounds
    
    \For{each candidate community $C \in C_{cand}$}{

        \If{$C.ub\_inf\_score_r <\theta$}{
        
            terminal the loop
            
        }
        
        compute the influence score $inf\_score_{C, Q} = \sf{calculate\_influence}$$(C, Q)$ 
        
        \If{$inf\_score_{C, Q} > \theta$}{
            
            add $C$ to $\mathbb{L}$

            remove a candidate community with the lowest influence score from $\mathbb{L}$

            update $\theta$ with the minimum influence score in $\mathbb{L}$
            
        }
    }
    \Return $\mathbb{L}$
\end{algorithm}

\noindent{\bf Initialization: } First, our Top\textit{M}-RICS algorithm obtains a query bit vector $L_q.BV$ by hashing all keywords from the query keyword set $L_q$ (line 1). Then, according to the given query center vertex $v_q$, the algorithm determines the target community $Q$ (line 2). After that, we initialize an empty list, $\mathbb{L}$, to store top-\textit{M} seed communities we have searched so far.  Moreover, we maintain an initially empty list, $C_{cand}$, which keeps a set of potential candidate communities for delayed refinement. We also set a variable, $\theta$, to $0$, which indicates the highest influence score we have encountered so far for the early termination of the index traversal (line 3). 



\noindent{\bf Index Traversal:} 
To facilitate the index traversal, we maintain a \textit{minimum heap} $\mathcal{H}$, which accepts heap entries in the form $(\mathcal{N}, key)$, where $\mathcal{N}$ is an index node, and $key$ is the minimum lower bound of the distances from vertices under node $\mathcal{N}$ to query vertex $v_q$ (line 4). To start the index traversal, we insert all entries in the root of index $\mathcal{I}$ into heap $\mathcal{H}$ (line 5). Then, we traverse the index by accessing entries from $\mathcal{H}$ in ascending order of distance lower bounds (intuitively, communities closer to $v_q$ will have higher influences on $Q$; lines 6-24). 


 Specifically, each time we pop out an index entry $(\mathcal{N}, key)$ with the minimum key from heap $\mathcal{H}$ (line 7). If $\theta \geq max\_inf\_ub(\mathcal{H})$ holds, which indicates that all the candidate communities in the remaining entries of $\mathcal{H}$ cannot have higher influences than the communities we have already obtained, then we can terminate the index traversal (lines 8-9); otherwise, we will check the entries in the node $\mathcal{N}$.



When $\mathcal{N}$ is a leaf node, for each vertex $v_i\in \mathcal{N}$, we first obtain its candidate community $C=r\text{-}hop(v_i,G)$ centered at $v_i$ (line 12).
Then, for the candidate community $C$, we apply the \textit{Keyword Pruning} (Lemma~\ref{lemma:keyword_pruning}), \textit{Support Pruning} (Lemma~\ref{lemma:support_pruning}), and \textit{Influence Score Pruning} (Lemma~\ref{lemma:influence_pruning}) (line 13).
If $C$ cannot be ruled out by these three pruning methods, we will check whether $C$ will be larger and closer to $Q$ than some community in $\mathbb{L}$, which intuitively may have higher influences on $Q$ (line 14). If the size of $\mathbb{L}$ is no more than $M$, we calculate the exact influence score, $inf\_score_{C, Q}$, from $C$ to target community $Q$, by invoking the function {\sf calculate\_influence}$(C, Q)$, and add $C$ to $\mathbb{L}$ as one of the current top-$M$ communities (lines 15-17). 
When we have obtained $M$ candidate communities in $\mathbb{L}$, we will update $\theta$ with the minimum influence score in $\mathbb{L}$ (line 18). If the influence score upper bound $v_i.ub\_inf\_score_r$ is higher than the score threshold $\theta$,  $C$ has the potential to have higher influence on $Q$, and thus will be added to $C_{cand}$ for later refinement (lines 19-21). If $C$ is not closer to $Q$ or not larger than some community in $\mathbb{L}$, we also add candidate community $C$ to $C_{cand}$ for later refinement (lines 22-23). 



When $\mathcal{N}$ is a non-leaf node, we will consider each child node $\mathcal{N}_i\in \mathcal{N}$ (lines 24-25). If entry $\mathcal{N}_i$ cannot be pruned by \textit{Index Keyword Pruning} (Lemma~\ref{Index Keyword Pruning}) and \textit{Index Support Pruning} (Lemma~\ref{Index Support Pruning}), we insert the entry $(\mathcal{N}_i,key)$ into heap $\mathcal{H}$ for further investigation (lines 26-27).

When either the heap $\mathcal{H}$ is empty (line 6) or the remaining index entries in $\mathcal{H}$ cannot contain candidate communities (line 8), we will terminate the index traversal.

\noindent{\bf Refinement of Candidate Communities:} After the index traversal, we update the list $C_{cand}$ by sorting candidate communities in descending order of influence score upper bounds (line 28). Then, for each candidate community $C$ in the list $C_{cand}$, if it holds that $C.ub\_inf\_score_r < \theta$, we can stop checking the remaining candidates in $C_{cand}$ (as all candidates in $C_{cand}$ have influence upper bounds less than the influence threshold $\theta$; lines 29-31). 
Next, we calculate the exact influence, $inf\_score_{C, Q}$, by invoking the function {\sf calculate\_influence}$(C, Q)$ (line 32).
If $inf\_score_{C, Q} > \theta$ holds, we need to add $C$ to the Top\textit{M}-RICS answer list $\mathbb{L}$, remove a candidate community with the lowest influence from $\mathbb{L}$, and update the influence threshold $\theta$ (lines 33-36). Finally, after refining $C_{cand}$, we return $\mathbb{L}$ as the Top\textit{M}-RICS answer list (line 37).

\noindent{\bf Discussions on the Computation of $\sf{calculate\_influence}(C,Q)$:} 
{\color{black}
To exactly calculate the community-to-community influence score (via Equation~(\ref{infSuv}) and (\ref{sigma_{C_i,C_j}})), we need to obtain the influence of each user in the seed community $C$ on the target community $Q$, the whole process is similar to the single-source shortest path algorithm. For each point $v_c$ in $C$, we first visit its 1-hop neighbors $v_b$ and the influence score $inf\_score_{vc, vb}=P_{vc, vb}$. Then, each time, we extend 1-hop neighbors $v_{new}$ forward and compute the current influence score $inf\_score_{vc,v_{new}}=\max_{\forall {v_i\in v_{new}}} (inf\_score_{v_{new},v_i} \cdot P_{v_{new},v_i})$, of $v_c$, until we get the maximum influence score on all node of $Q$.
}

{\color{black} \noindent{\bf Discussions on the Online Computation of Influence Upper Bound $v_i.ub\_inf\_score_r$:} Since we get the upper bound of boundary influence score, $v_i.ub\_bound\_inf_r$, of $\sf{collapse\_calculate}$ data for a subgraph $r\text{-}hop(v_i, G)$, and for a target community, $Q$, with query center vertex, $v_q$, we can get the distance lower bound $lb\_dist(v_i, v_q) = \min([|v_i.Dist[j] - v_q.Dist[j]|]$, for $1 \leq j \leq d)$ between $v_i$ and $v_q$ by \textit{triangle inequality} \cite{Plaisted84}. Then, we can get the upper bound of influence score, $v_i.ub\_inf\_score_r = v_i.ub\_bound\_inf_r \cdot |V(Q)| \cdot \max(P)^{lb\_dist(v_i, v_q)-2\cdot r}$, where $\max(P)$ denotes the maximum neighbor activation probability in $G$.}

{\color{black}
\noindent{\bf Complexity Analysis:} Let $\overline{n_r}$ be the average number of users in the target community $Q$. The cost of obtaining $L_q.BV$ and $Q$ takes $O(\overline{n_r})$. Let $PP_j$ be the pruning power (i.e., the percentage of node entries that can be pruned) on the $j$-th level of the tree index $\mathcal{I}$, where $0 \leq j \leq h$ and $h$ is the height of the tree. Denote $f$ as the average fanout of nodes in index $\mathcal{I}$. For the index traversal, the number of visited nodes is given by $O(\sum_{j=1}^h f_{h-j+1} \cdot (1-PP_j))$. We label a subgraph $r\text{-}hop(v_i, G)$ as $g$. Each time, the function of $\sf{calculate\_influence}$ need $O((|V(g)|+|E(g)|) \cdot \overline{n_r})$. Let $\overline{n_d}$ be the average number of iterations updated due to the closest distance. Then, the updating $\mathbb{L}$ and $\theta$ takes $O(M \cdot (|V(g)|+|E(g)|) \cdot \overline{n_r} \cdot \overline{n_d})$. And, updating $C_{cand}$ takes $O(1)$. For the refinement process, let $\overline{n_m}$ be the average number of calculate influence in $C_{cand}$, and the updating $\mathbb{L}$ and $\theta$ take $O(1)$. Therefore, the total time complexity of Algorithm \ref{alg: online_RICS} is given by $O(\sum_{j=1}^h f_{h-j+1} \cdot (1-PP_j) + ((|V(g)|+|E(g)|) \cdot (M \cdot \overline{n_d} + \overline{n_m}) + 1) \cdot \overline{n_r})$.
}

\section{A Variant of Top\textit{M}-RICS}
\label{sec: extension}

In this section, we formulate and tackle a variant of Top$M$-RICS which relaxes structural constraints for real applications.

\noindent {\bf A Variant, Top$M$-R$^2$ICS, of the Top$M$-RICS Problem:} In Definition \ref{def:RICS}, our Top\textit{M}-RICS problem returns top-$M$ most influential seed communities, where a seed community needs to fulfill structural requirements (e.g., $k$-truss and radius constraint, as given in Definition \ref{def:seed_community}). In this paper, we also consider its variant, named \textit{\textbf{Top-$M$ R}elaxed \textbf{R}everse \textbf{I}nfluential \textbf{C}ommunity \textbf{S}earch} (Top\textit{M}-R$^2$ICS), which obtains top-\textit{M} communities with the relaxed structural constraints and having the highest influences.

\begin{definition}
    (\textbf{Top-$M$ Relaxed Reverse Influential Community Search, Top$M$-R$^2$ICS})
    \label{def:R2ICS}
    Given a social network $G=(V(G), E(G),\Phi(G))$, a set, $L_q$, of query keywords, the maximum number, $N$, of community users, and a target community $Q$ (with center vertex $v_q$, radius $r$, and query keywords in $L_q$), the problem of the top-M relaxed reverse influential community search (TopM-R$^2$ICS) retrieves a list of $M$ subgraphs, $S_l$ (for $1 \leq l \leq M$), from the social network, $G$, such that:
    \begin{itemize}
        \item each subgraph $S_l \subseteq G$ satisfies the size constraint that $|V(S_l)| \leq N$;
        \item for any vertex $v_i \in V(S_l)$, its keyword set $v_i.L$ contains at least one query keyword in $L_q$ (i.e., $v_i.L \cap L_q$$\neq$$\emptyset$), and;
        \item $M$ subgraphs $S_l$ have the highest community-to-community influences, $\mathit{inf\_score}_{S_l, Q}$.
    \end{itemize}
\end{definition}

In Definition \ref{def:R2ICS}, the variant, Top\textit{M}-R$^2$ICS, retrieves top-$M$ subgraphs, $S_l$, without structural constraints such as $k$-truss, radius $r$, and the connectivity, as used in the Top\textit{M}-RICS problem (as given in Definition \ref{def:RICS}).

\noindent {\bf Effective Pruning Strategy w.r.t. Vertex-to-Community Influence Score:} In order to obtain Top\textit{M}-R$^2$ICS community answer $S_l$ in $G$ with the highest influence $inf\_score_{S_l,Q}$, a straightforward method is to enumerate all the vertex combinations and find a community that meets the requirements, which is however rather inefficient. Instead, we observed that the community-to-community influence score $inf\_score_{S_l,Q}$ (as given in Eq.~(\ref{sigma_{C_i,C_j}})) is given by summing up vertex-to-community influences, $inf\_score_{v, Q}$, for all $v\in V(S_l)$. Based on this observation, our basic idea about the Top\textit{M}-R$^2$ICS algorithm is to retrieve those vertices $v$ with high vertex-to-community influences on $Q$ first and early terminate the search to prune or avoid accessing those low-influence vertices. 

Therefore, to enable the pruning of low-influence vertices, we propose an effective \textit{vertex-to-community influence score pruning} strategy as follows.

\begin{lemma}
    \label{lemma:R^2ICS influence}
    \textbf{(Vertex-to-Community Influence Score Pruning)} Let $\theta$ be the $N$-th highest vertex-to-community influence from Top$M$-R$^2$ICS candidate vertices we have obtained so far to the target community $Q$. Any vertex $v$ can be safely pruned, if it holds that $ub\_inf\_score_{v,Q} < \theta$, where $ub\_inf\_score_{v,Q}$ is an upper bound of the influence score $inf\_score_{v,Q}$.
\end{lemma}
\begin{proof}
    Since $ub\_inf\_score_{v,Q}$ is an upper bound of the influence score $inf\_score_{v,Q}$, we have $inf\_score_{v,Q} \geq ub\_inf\_score_{v,Q}$. If $ub\_inf\_score_{v,Q} < \theta$ holds, means $inf\_score_{v,Q} < \theta$,  which indicates that the vertex $v$ has a lower influence on $Q$, compared with some vertices we have obtained so far (i.e., with influence threshold $\theta$), and $v$ cannot be one of our TopM-R$^2$ICS answers. Therefore, we can safely prune candidate vertex $v$, which completes the proof.
\end{proof}

\noindent {\bf A Framework for the  Top\textit{M}-R$^2$ICS Algorithm:} Algorithm \ref{alg: onlineR^2ICS} illustrates the pseudo code of our Top\textit{M}-R$^2$ICS algorithm over a social network $G$, which consists of \textit{initialization}, \textit{index traversal}, and \textit{candidate vertex refinement} phases. We first initialize the data/variables that will be used in Algorithm \ref{alg: onlineR^2ICS} (lines 1-3). Then, for each user-specified Top\textit{M}-R$^2$ICS query, we traverse the index $\mathcal{I}$ to obtain candidate vertices, by applying our proposed keyword pruning strategies in Sections~\ref{keyword_pruning_4_1} and \ref{Index Pruning}  (lines 4-20). Finally, we refine candidate vertices and return actual Top\textit{M}-R$^2$ICS subgraph list, $\mathbb{L}$, with the highest influence scores on target community $Q$ (lines 21-36).




\underline{\it Initialization:} In the initialization phase, the algorithm first hashes all the query keywords in $L_q$ into a query keyword bit vector $L_q.BV$, and then obtain the target community $Q$ (i.e., $r$-hop subgraph with center vertex $v_q$) (lines 1-2). Next, we prepare an initially empty list $\mathbb{L}$ to store the top-\textit{M} communities, an initially empty set, $v2C\_inf\_set$, to store vertex-to-community influence scores of each vertex in $\mathbb{L}$, an initially empty set, $V_{vis}$, to store the vertices we have searched so far, and an empty candidate vertex set, $V_{cand}$, for later refinement (line 3).

\underline{\it Index Traversal:} We utilize a \textit{minimum heap} $\mathcal{H}$ with heap entries ($\mathcal{N}, key$) for the index traversal, where $\mathcal{N}$ is an index node, and $key$ is defined as the lower bound of distances from vertices under node $\mathcal{N}$ to query vertex $v_q$ (line 4). To start the index traversal, we insert all entries $\mathcal{N}$ in the root of index $\mathcal{I}$ into heap $\mathcal{H}$ (line 5). 


Each time we pop out an index entry $(\mathcal{N}, key)$ with the minimum key from heap $\mathcal{H}$ (lines 6-7). When $\mathcal{N}$ is a leaf node, for each vertex $v_i \in \mathcal{N}$, we apply the \textit{keyword pruning} (Lemma \ref{lemma:keyword_pruning}) (lines 8-10). If $v_i$ cannot be ruled out by this pruning method, we will decide whether we add $v_i$ to $V_{vis}$ or $V_{cand}$ (lines 11-16). 
If $v_i\in \mathcal{N}$ is closer to $v_q$ than all vertices in $V_{vis}$ (i.e., potentially having higher influence on $Q$), we will compute the exact vertex-to-community influence score $inf\_score_{v_i, Q}$, add vertex $v_i$ to $V_{vis}$, and update set $v2C\_inf\_set$ by adding $inf\_score_{v_i, Q}$ (lines 11-14); otherwise (i.e., $v_i$ is far away from $v_q$), we will add vertex $v_i$ to the vertex candidate set $V_{cand}$ for later refinement (lines 15-16).


When $\mathcal{N}$ is a non-leaf node, we will consider each child node $\mathcal{N}_i \in \mathcal{N}$ (lines 17-18). If $\mathcal{N}_i$ cannot be pruned by Lemma~\ref{Index Keyword Pruning}, then we insert entry ($\mathcal{N}_i, key$) into heap $\mathcal{H}$ for further investigation (lines 19-20).

\underline{\it Candidate Vertex Refinement:} Next, we will further check candidate vertices in $V_{cand}$ and form the final subgraph answers in $\mathbb{L}$ with top-$M$ highest influences. 
Specifically, we first sort candidate vertices in $V_{cand}$ in descending order of influence score upper bounds, and let threshold $\theta$ be the minimum influence score in $v2C\_inf\_set$ we have seen so far (lines 21-22). 
Then, for each candidate vertex $v_i \in V_{cand}$, if we have obtained $M$ subgraphs in $\mathbb{L}$, then we can terminate the loop early (lines 34-35); otherwise, we consider the remaining vertices (lines 23-33).

If the visited vertex set $V_{vis}$ reaches the maximum size $N$ and the vertex-to-community score upper bound  $v_i.ub\_inf\_score_r$ is greater than or equal to $\theta$, we can obtain a candidate subgraph $S_l$ consisting of top-$N$ vertices in $v2C\_inf\_set$ with the highest influence scores and add this candidate subgraph $S_l$ to $\mathbb{L}$. Next, we update $V_{vis}$ and $v2C\_inf\_set$ (i.e., remove the vertex of $S_l$ with the highest influence score in $V_{vis}$ and $v2C\_inf\_set$; lines 25-29).
After that, we obtain the exact vertex-to-community influence score $inf\_score_{v_i, Q} = \sf{calculate\_influence}$$(v_i, Q)$, add $v_i$ to $V_{vis}$, add $inf\_score_{v_i, Q}$ to $v2C\_inf\_set$, and update the threshold $\theta$ with $\min(v2C\_inf\_set)$ (lines 30-33). This process continues until we obtain $M$ candidate subgraphs in $\mathbb{L}$ (lines 34-35).


Finally, we return $M$ communities in $\mathbb{L}$ as the Top$M$-R$^2$ICS answer (line 36).

\noindent{\bf Discussions on the Differences from the Top$M$-RICS Processing Algorithm:} 
Different from Top$M$-RICS, Top\textit{M}-R$^2$ICS considers vertex-to-community influence $inf\_score_{v_i, Q}$ on $Q$ (instead of the community-to-community influence). Thus, our Top\textit{M}-R$^2$ICS algorithm aims to obtain $M$ sets of top-$N$ vertices with the highest influence scores on the target community $Q$, which form $M$ communities, respectively, as our Top\textit{M}-R$^2$ICS answers. Moreover, our Top$M$-RICS approach applies various pruning strategies (i.e., Lemmas~\ref{lemma:keyword_pruning}, \ref{lemma:support_pruning}, and \ref{lemma:influence_pruning}) to reduce the community search space, and retrieves top-$M$ most influential communities. In contrast, since Top$M$-R$^2$ICS does not require structural constraints, Algorithm \ref{alg: onlineR^2ICS} only uses keyword and vertex-to-community influence score pruning (i.e., Lemmas~\ref{lemma:keyword_pruning} and \ref{lemma:R^2ICS influence}, resp.) to filter out false alarms. Further, our Top$M$-R$^2$ICS approach considers the vertex-to-community influences of individual vertices in subgraphs with high influences (rather than the community-level retrieval in Top$M$-RICS).

\noindent
\textbf{Discussions on the Correctness of the Top\textit{M}-R$^2$ICS Algorithm:} Since we only use the keyword pruning during the index traversal phase, we can obtain three sets: $V_{vis}$, $V_{cand}$, and $v2C\_inf\_set$, containing all candidate vertices (and their influence set) that meet the keyword requirement. Thus, we do not miss any vertices with high influences on $Q$ in this step.
In Definition~\ref{def:c2c_influence}, the community-to-community influence is given by the summation of vertex-to-community influences. From our Top\textit{M}-R$^2$ICS algorithm in Algorithm \ref{alg: onlineR^2ICS}), we always include in $S_l$ of $\mathbb{L}$ those vertices $v_i$ with the highest vertex-to-community influences $inf\_score_{v_i, Q}$ (lines 26-29). Thus, the Top$M$-R$^2$ICS results are guaranteed to achieve the highest community-to-community influence scores.

{\color{black}


\begin{algorithm}[!htb]\small
\caption{{\bf Online Top\textit{M}-R$^2$ICS Processing}\small}
\label{alg: onlineR^2ICS}\footnotesize
\KwIn{
    \romannumeral1) a set, $L_{q}$, of query keywords,
    \romannumeral2) a query center vertex, $v_q$,
    \romannumeral3) the maximum radius, $r$, of the target community $Q$,
    \romannumeral4) the maximum number, $N$, of users in the R$^2$ICS community,
    \romannumeral5) the index $\mathcal{I}$ over social networks $G$, and
    \romannumeral6) an integer parameter, $M$
}
\KwOut{
    a list, $\mathbb{L}$, of Top\textit{M}-R$^2$ICS communities answer
}

    \tcp{initialization phase}

    hash all keywords in the query keyword set $L_{q}$ into a query bit vector $L_{q}. BV$

    obtain the target community $Q = r\text{-}hop(v_q, G)$

    $\mathbb{L}=\emptyset, v2C\_inf\_set=\emptyset, V_{vis}=\emptyset, V_{cand}=\emptyset$\;

    \tcp{index traversal phase}

    initialize a minimum heap $\mathcal{H}$ accepting index entries in the form $(\mathcal{N},key)$

    insert all entries $\mathcal{N}$ in the root of index $\mathcal{I}$ into heap $\mathcal{H}$

    \While{$\mathcal{H}$ is not empty}{

        $(\mathcal{N},key)=\mathcal{H}.pop()$

    
        \eIf{$\mathcal{N}$ is a leaf node}{
            \For{each vertex $v_i\in \mathcal{N}$}{

                \If{ $v_i$ cannot be pruned by Lemma~\ref{lemma:keyword_pruning}}{

                \eIf{$v_i$ is closer to $v_q$ than all vertices in $V_{vis}$}{
                    
                    compute the influence score $inf\_score_{v_i, Q}$ $=$ $\sf{calculate\_influence}$$(v_i, Q)$

                    add $v_i$ to $V_{vis}$
                    
                    add $inf\_score_{v_i, Q}$ to $v2C\_inf\_set$

                }{
                    add $v_i$ to $V_{cand}$
                }
                
                }
            }
        }
        {   
        \tcp{$\mathcal{N}$ is a non-leaf node}
            \For{each entry $\mathcal{N}_i \in \mathcal{N}$}{
                \If{$\mathcal{N}_i$ cannot be pruned by Lemma~\ref{Index Keyword Pruning}}{
                insert $(\mathcal{N}_i,key)$ into heap $\mathcal{H}$ 
                }
            }
        }
    }

    \tcp{candidate vertex refinement phase}
    
    sort candidate vertices in $V_{cand}$ in descending order of influence score upper bounds

    $\theta = \min(v2C\_inf\_set)$
    
    \For{each candidate vertex $v_i \in V_{cand}$}{

        \eIf{$|\mathbb{L}| < M$}{

            \If{$|V(V_{vis})| \geq N$ and $v_i.ub\_inf\_score_r \geq \theta$}{

                obtain a subgraph $S_l$ consisting of top-$N$ vertices in $v2C\_inf\_set$ with the highest influence score

                add $S_l$ to $\mathbb{L}$

                remove the vertex with the highest influence score in $V_{vis}$

                remove the highest influence score in $v2C\_inf\_set$
                
        
            }
            
            compute the influence score $inf\_score_{v_i, Q} = \sf{calculate\_influence}$$(v_i, Q)$

            add $v_i$ to $V_{vis}$ 
                
            add $inf\_score_{v_i, Q}$ to $v2C\_inf\_set$

            update $\theta = \min(v2C\_inf\_set)$
            
        }{
            terminate the loop
        }
    }
    \Return $\mathbb{L}$
\end{algorithm}

\noindent
\textbf{Complexity Analysis:} As the same Algorithm~\ref{alg: online_RICS}, the initialization and index traversal phase of our Top\textit{M}-R$^2$ICS takes $O(\sum_{j=1}^h f_{h-j+1} + (|V(g)|+|E(g)|+1) \cdot \overline{n_r} \cdot N \cdot M)$. For refinement, let $GPP$ be the vertex-to-community influence score pruning power (i.e., the percentage of vertices that can be pruned). And, updating $\mathbb{L}$, $v2C\_inf\_set$ and $\theta$ takes $O(1)$. Therefore, the total time complexity of Algorithm~\ref{alg: onlineR^2ICS} is given by $O(\sum_{j=1}^h f_{h-j+1} + ((|V(g)|+|E(g)|) \cdot (N + |V(G)| \cdot GPP) + 1) \cdot \overline{n_r} \cdot M)$

\section{Experimental evaluation}
\label{Experiment}
\subsection{Experimental Settings}
We evaluate the performance of the online Top\textit{M}-RICS algorithm (i.e., Algorithm \ref{alg: online_RICS}) on both real and synthetic graph data sets.


\noindent{\bf Real-World Graph Data Sets:} We use three real-world graphs, Facebook \cite{leskovec2012learning}, Amazon \cite{yang2012defining}, and DBLP \cite{zhou2009graph}, whose statistics are depicted in Table \ref{tab:data_set}. Facebook is a social network, where two users are connected if they are friends. Amazon is an Also Bought network, where two products are connected if they are purchased together. DBLP is a co-authorship network, where two authors are connected if they publish at least one paper together.

\noindent{\bf Synthetic Graph Data Sets:} We construct synthetic social networks by generating \textit{small-world graphs} $G$ \cite{newman1999renormalization}. Specifically, we first create a ring of size $|V(G)|$, and then connect $m$ nearest neighbor nodes for each vertex $u$. Next, for each generated edge $e_{u,v}$, we add a new edge $e_{u,w}$ with probability $\mu$ that connects $u$ to a random vertex $w$. Here, we take $m$ = 5 and $\mu$ = 0.251. For each vertex, we randomly generate a keyword set $v_i.L$ from the keyword domain $\Sigma$, following $Uniform$, $Gaussian$, and $Zipf$ distributions, to obtain three synthetic graphs, denoted as $Uni$, $Gau$, and $Zipf$, respectively. Next, for each edge $e_{u,v}$ in the generated graphs, we produce a random value within an interval $[0.5,0.6)$ as the edge activation probability $P_{u,v}$.

\begin{table}[t]
\begin{center}
\caption{Statistics of the tested real-world graph data sets.}
\label{tab:data_set}
\footnotesize
\begin{tabular}{|l||l||l|}
\hline
\textbf{Social Networks} & $|V(G)|$ & $|E(G)|$\\
\hline\hline
    \textit{Facebook\cite{leskovec2012learning}} & 4,039 & 88,234\\\hline
    \textit{Amazon\cite{yang2012defining}} & 334,863 & 925,872\\\hline
    \textit{DBLP\cite{zhou2009graph}} & 317,080 & 1,049,866\\\hline
\end{tabular}
\end{center}
\vspace{-1ex}
\end{table}

\begin{table}[t]
\begin{center}
\vspace{-1ex}
\caption{Parameter settings.}
\label{tab:parameters}\vspace{-2ex}
\footnotesize
\begin{tabular}{|l||l|}
\hline
\textbf{Parameters}&\textbf{Values} \\
\hline\hline
    support, $k$, of truss structure & 3, \textbf{4}, 5\\\hline    
    radius $r$ & 1, \textbf{2}, 3\\\hline
    size, $|L_q|$, of query keywords set & 2, 3, {\bf 5}, 8, 10 \\\hline
    size, $|v_i.L|$, of keywords per vertex & 1, 2, {\bf 3}, 4, 5 \\\hline
    size, $M$, of query result set & {\bf 1}, 2, 3, 4, 5 \\\hline
    keyword domain size $|\Sigma|$ & {\color{black}10, {\bf 20}, 50, 80}\\\hline 
    the number, $d$, of pivots & 3, {\bf 5}, 8, 10\\\hline
    the maximum size, $N$, of seed community& {\color{black}5, {\bf 10}, 15, 20}\\\hline
    the size, $|V(G)|$, of data graph $G$ & 10K, 25K, {\bf 50K}, 100K, 250K\\\hline
\end{tabular}
\end{center}\vspace{-2ex}
\end{table}

\noindent{\bf Competitors:} To our best knowledge, no prior works studied the Top\textit{M}-RICS problem {\color{black}and its variant, Top\textit{M}-R$^2$ICS,} by considering the influence of a connected community (or disconnected subgraph) on a user-specified target community (instead of the entire graph). Therefore, we compare our Top\textit{M}-RICS approach with a straightforward method, called \textit{baseline}. The \textit{baseline} method first determines the target community $Q$ based on the given query vertex and then performs \textit{Breadth First Search} (BFS) from $Q$ in the social network $G$.
For each vertex we encounter (during the BFS traversal), \textit{baseline} obtains its $r$-hop subgraph and checks whether this subgraph satisfies the structure and keyword constraints.
Next, we obtain candidate communities $S_l$ from the $r$-hop subgraph and calculate their influence scores $inf\_score_{S_l, Q}$. If a candidate community $S_l$ has an influence score greater than the best score we have seen so far, we will let $S_l$ be the best-so-far Top\textit{M}-RICS answer.
Finally, after all vertices have been traversed, \textit{baseline} returns top-\textit{M} candidate communities we have obtained with the maximum influence on the target community $Q$.
Note that, since the time cost of the \textit{baseline} method is extremely high, we evaluate this method by sampling $0.1\%$ vertices from the data graph $G$ without replacement. Therefore, the total time can be estimated by $\overline{t} \cdot |V(G)|$, where $\overline{t}$ denotes the average time of each sample.

{\color{black} For Top\textit{M}-R$^2$ICS, we compare our approach (Top\textit{M}-R$^2$ICS, Algorithm~\ref{alg: onlineR^2ICS}) with Top\textit{M}-R$^2$ICS\_WoP and \textit{Optimal} methods. Here, Top\textit{M}-R$^2$ICS\_WoP is our Top\textit{M}-R$^2$ICS approach without the pruning strategy in Section~\ref{sec: extension}, whereas \textit{Optimal} computes the influence score of each vertex on the target community in the original graph and selects a combination of $N$ vertic es with the highest influence score as the query result.}

\noindent{\bf Measure:} {\color{black}To evaluate the efficiency of our Top\textit{M}-RICS / Top\textit{M}-R$^2$ICS approaches, we randomly select 50 query nodes from each graph data set, and take the average of the \textit{wall clock time} over 50 runs}, which is the time cost of online retrieving Top\textit{M}-RICS or Top\textit{M}-R$^2$ICS query results via the index {\color{black}(Algorithms \ref{alg: online_RICS} and \ref{alg: onlineR^2ICS}, respectively).} 

\begin{figure}[t]
    \centering
    \includegraphics[width=0.30\textwidth]{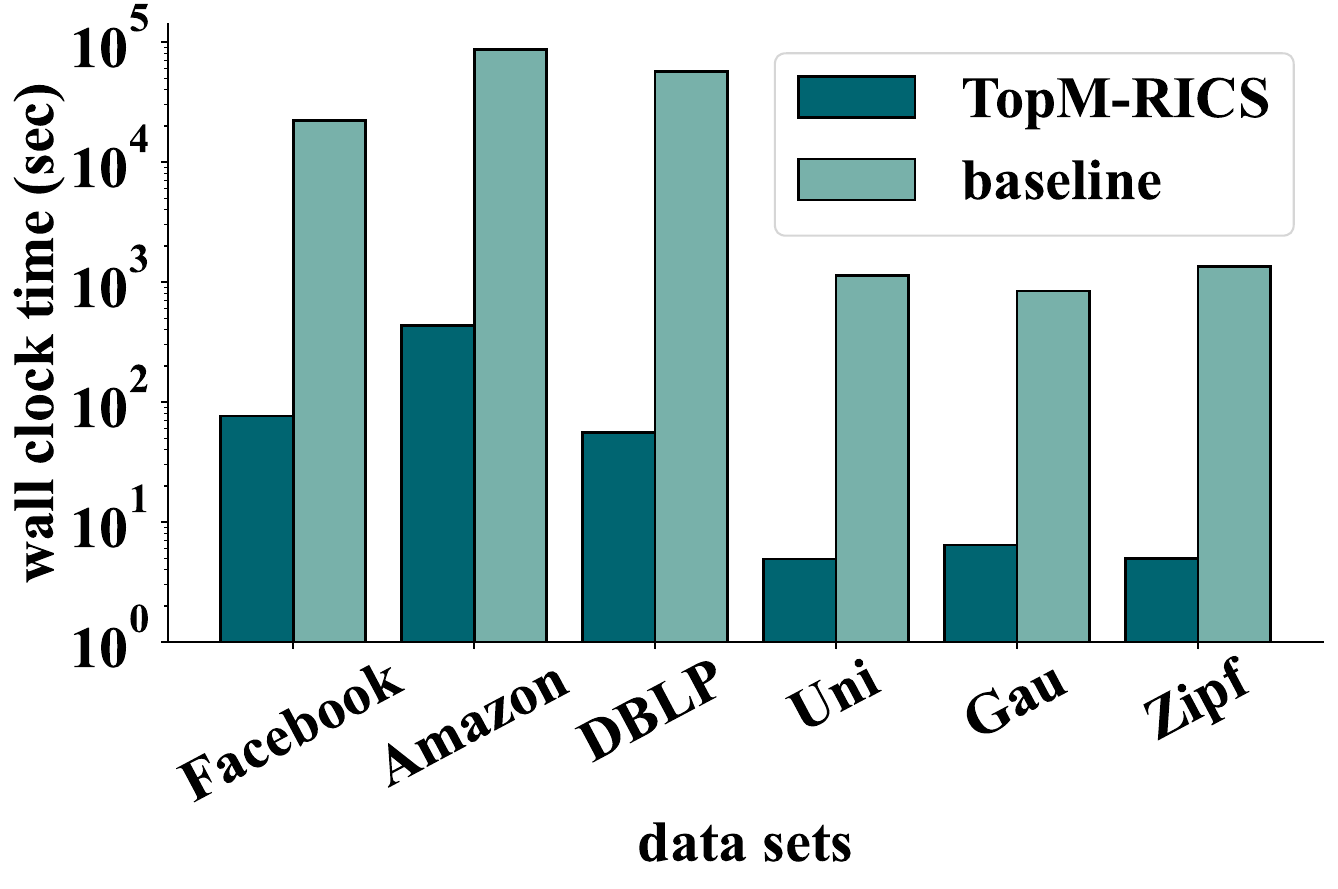}\vspace{-2ex}
    \caption{The Top\textit{M}-RICS performance on real/synthetic graphs.}
    \label{fig:per_out}\vspace{-2ex}
\end{figure}

\begin{figure*}[t]
    \centering
    \subfigure[edge support threshold, $k$]{\includegraphics[width=0.25\textwidth]{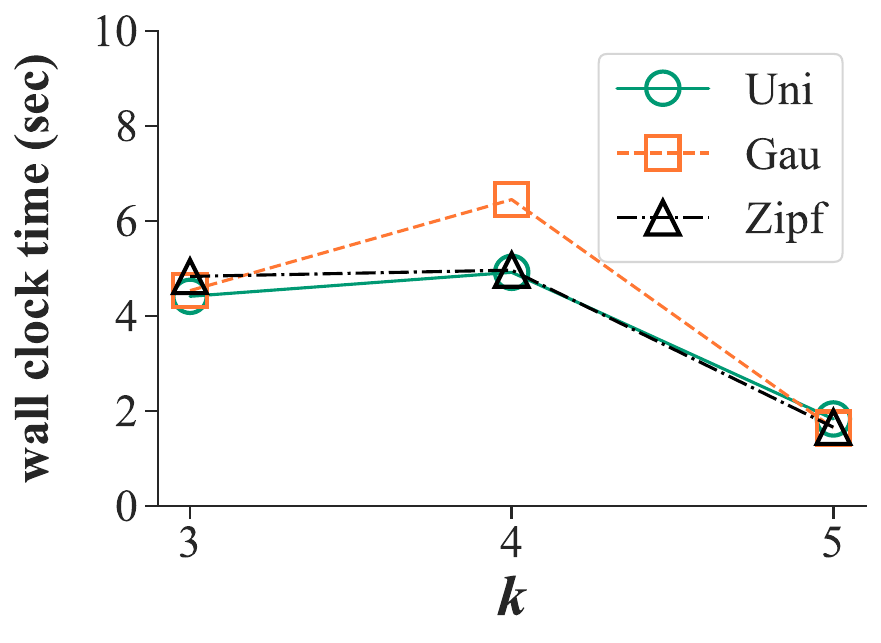}}\hfill
    \subfigure[radius, $r$]{\includegraphics[width=0.25\textwidth]{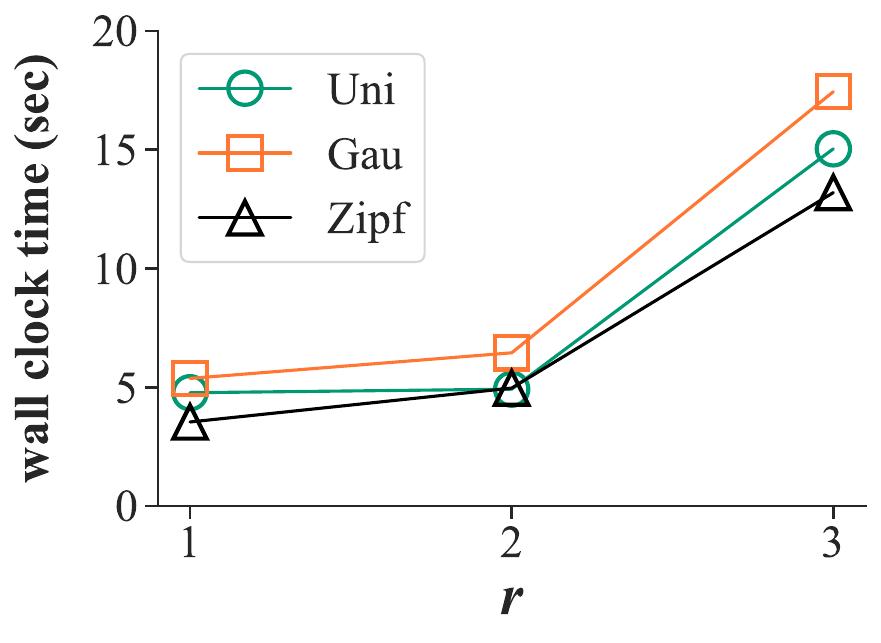}}\hfill
    \subfigure[$\#$ of query keywords, $|L_q|$]{\includegraphics[width=0.25\textwidth]{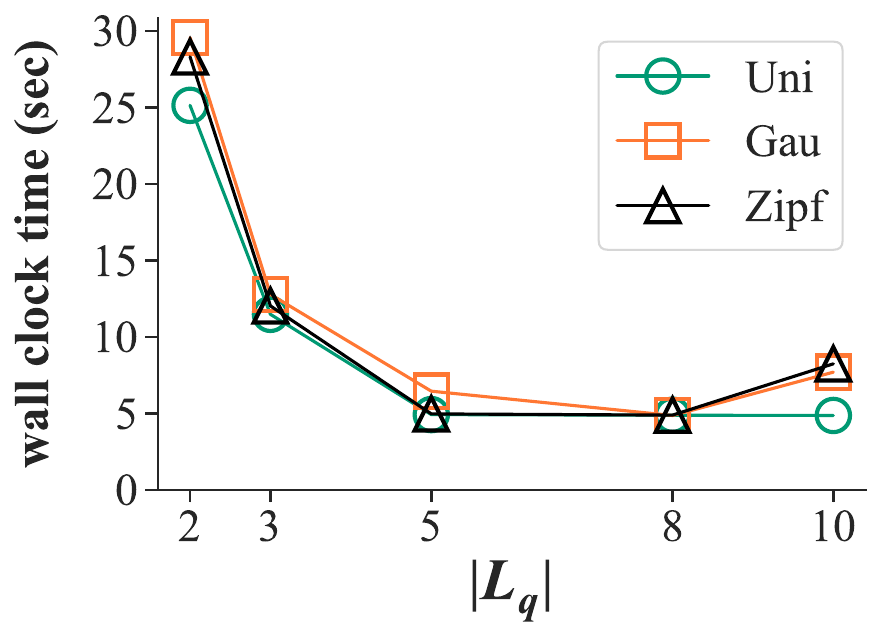}}\hfill
    \subfigure[$\#$ of keywords per vertex, $|v_i.L|$]{\includegraphics[width=0.25\textwidth]{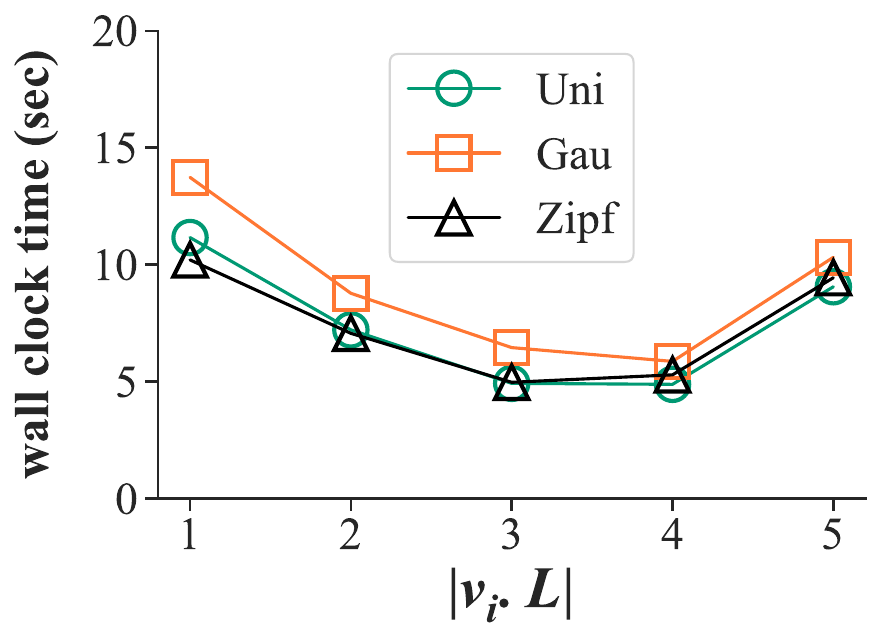}}\\\vspace{-1ex} 
    \subfigure[query result size, $M$]
    {\includegraphics[width=0.25\textwidth]{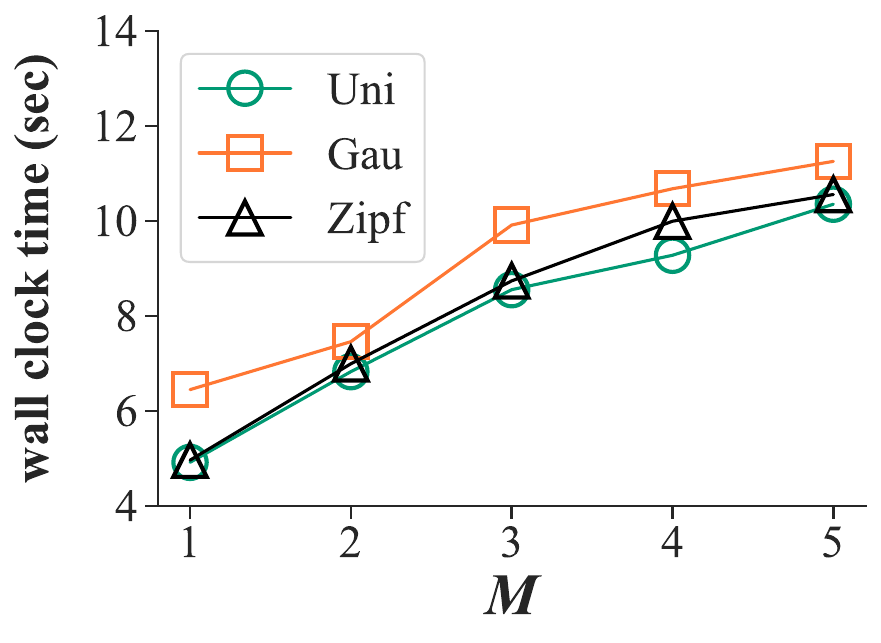}}\hfill
    \subfigure[$\#$ of pivots, $d$]{\includegraphics[width=0.25\textwidth]{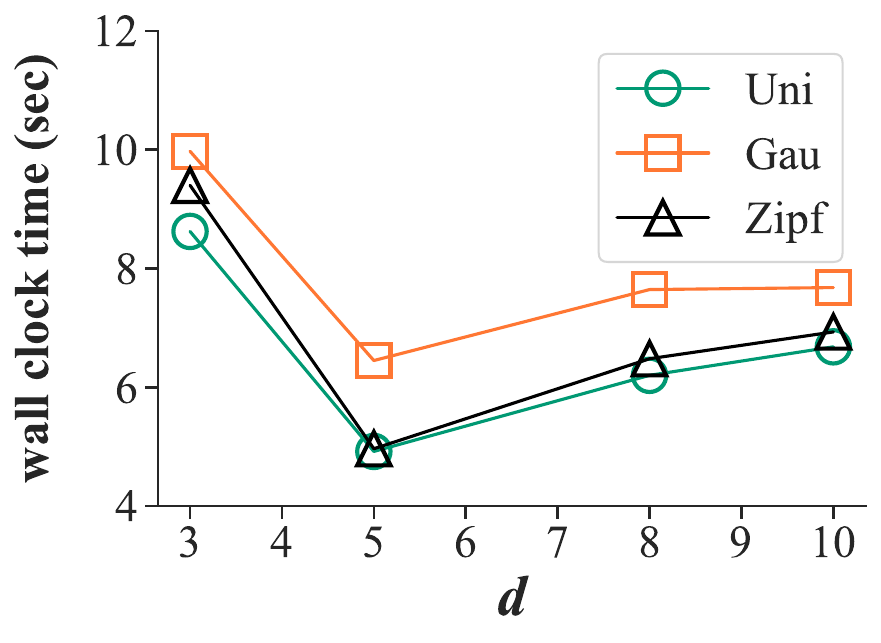}}\hfill
    \subfigure[maximum seed community size, $N$]{\includegraphics[width=0.25\textwidth]{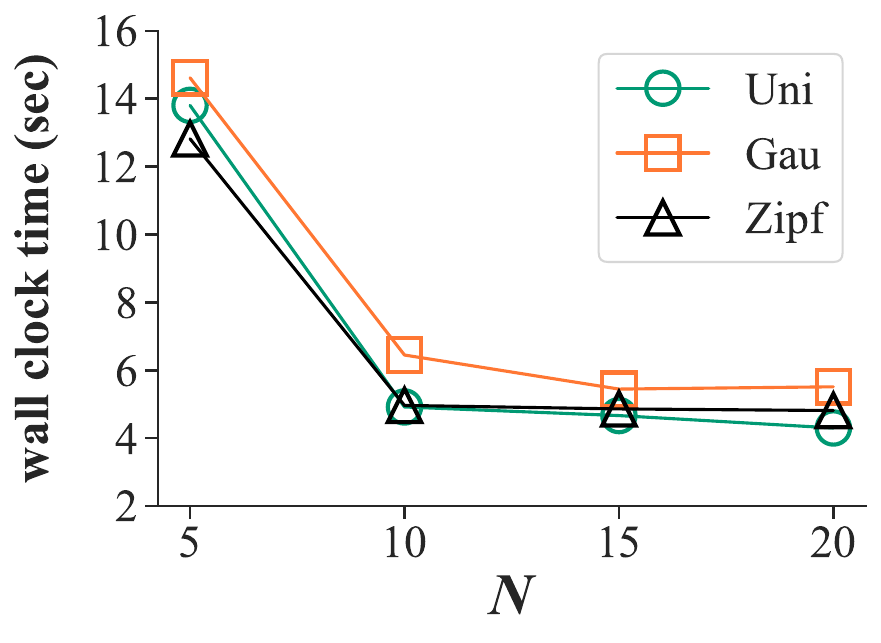}}\hfill
    \subfigure[graph size, $|V(G)|$]{\includegraphics[width=0.25\textwidth]{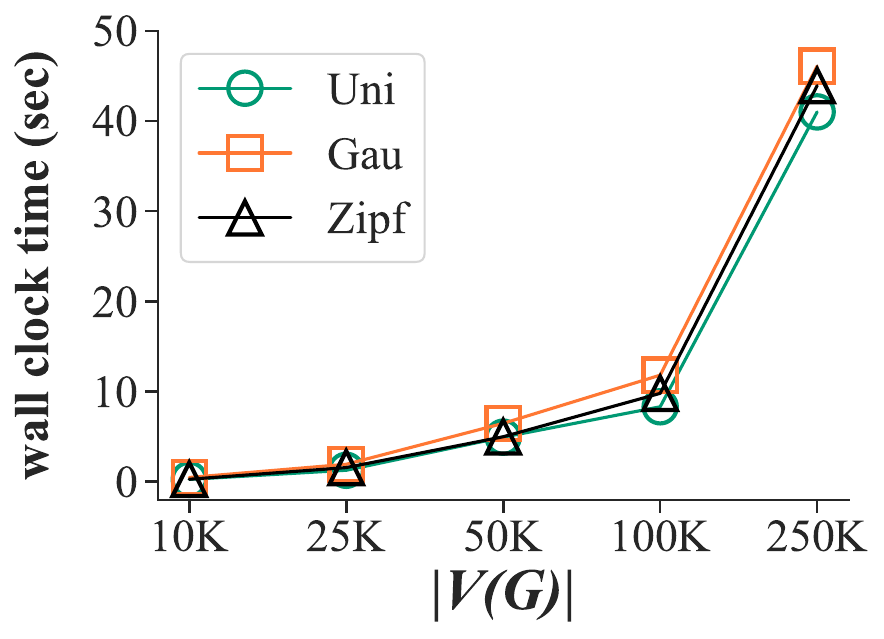}}\par
    \caption{The robustness evaluation of the Top\textit{M}-RICS query performance.}
    \label{fig:subfigures}\vspace{-2ex}
\end{figure*}

\begin{figure}
    \centering
    \includegraphics[width=0.5\linewidth]{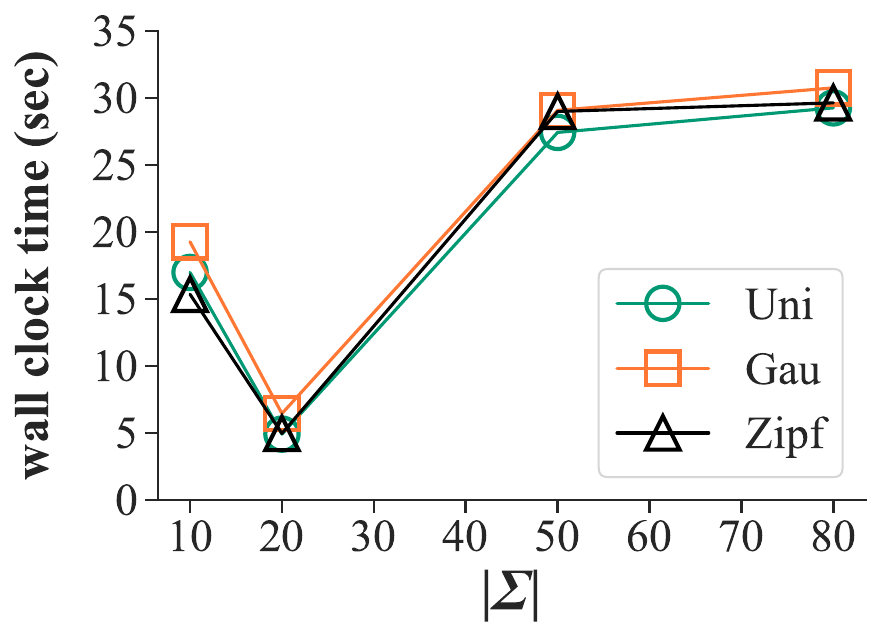}
    \caption{keyword domain size, $|\Sigma|$ of Top\textit{M}-RICS performance}
    \label{fig:sigma}\vspace{-2ex}
\end{figure}

\noindent{\bf Parameter Settings:} Table \ref{tab:parameters} depicts parameter settings, where default values are in bold. Each time, we vary the value of one parameter while setting other parameters to their default values. We ran all the experiments on the PC with Intel(R) Core(TM) i9-10900K CPU  3.70GHz and 32 GB memory. All algorithms were implemented in Python and executed with Python 3.8 interpreter.

\subsection{Top\textit{M}-RICS Performance Evaluation}
\noindent{\bf The Top\textit{M}-RICS Performance on Real/Synthetic Graphs:} Figure \ref{fig:per_out} illustrates the performance of our Top\textit{M}-RICS approach on both real and synthetic graphs, compared with the \textit{baseline} method, where all parameters are set by their default values in Table \ref{tab:parameters}, except for the dense \textit{Facebook} dataset with the maximum seed community size $N$ = 700. Experimental results show that the \textit{wall clock time} of our Top\textit{M}-RICS approach outperforms \textit{baseline} by almost three orders of magnitude, which confirms the effectiveness of our proposed pruning strategies and indexing mechanisms and the efficiency of our Top\textit{M}-RICS approach.

To evaluate the robustness of our Top\textit{M}-RICS approach, in subsequent experiments, we will test the effect of each parameter in Table \ref{tab:parameters} on the query performance over synthetic graphs.

\noindent{\bf Effect of Truss Support Parameter $k$:} Figure \ref{fig:subfigures}(a) shows the Top\textit{M}-RICS query performance for different $k$ values, where $k$ = 3, 4, and 5, and the rest of parameters are set to default values. From this figure, we can find that for larger $k$ values, the query time cost decreases over all three synthetic graphs. This is because larger $k$ leads to fewer candidate communities satisfying the $k$-truss constraints, which in turn incurs lower wall clock time.

\noindent{\bf Effect of Radius $\boldsymbol{r}$:} Figure \ref{fig:subfigures}(b) illustrates the wall clock time of our Top\textit{M}-RICS approach, by varying $r$ from 1 to 3, where other parameters are set to their default values. When the radius $r$ increases, the numbers of vertices included in the target and seed communities also increase, leading to higher filtering and refinement costs. Nevertheless, the wall clock time remains low (i.e., $3.54 \sim 17.44$ $sec$) for all the synthetic graphs.

\noindent{\bf Effect of the Size, $\boldsymbol{|L_q|}$, of the Query Keyword Set:} Figure \ref{fig:subfigures}(c) presents the Top\textit{M}-RICS query performance, where $|L_q|$ = 2, 3, 5, 8, and 10, and other parameters are by default. Intuitively, as $|L_q|$ increases, more candidate seed communities satisfy the keyword requirements. Thus, we will have a higher threshold of the influence score, which results in higher pruning power and, in turn, lower time cost, as confirmed by the figure. However, more candidate seed communities require more refinement costs, and wall clock times increase. In summary, the wall clock time remains low for different $|L_q|$ values (i.e., $4.86 \sim 29.56$ $sec$).

\noindent{\bf Effect of the Size, $\boldsymbol{|v_i.L|}$, of Keywords per vertex:} Figure \ref{fig:subfigures}(d) reports the efficiency of our Top\textit{M}-RICS approach, by varying $|v_i.L|$ from 1 to 5, where default values are used for other parameters. With the increase of $|v_i.L|$, more vertices are likely to be included in candidate seed communities, which leads to a higher influence threshold and higher pruning power (or lower query cost). Meanwhile, larger $|v_i.L|$ will incurs higher filtering/refinement costs. Therefore, the two factors mentioned above show that the wall clock time first decreases and then increases for larger $|v_i.L|$. The wall clock times with different $|v_i.L|$ values are $4.88 \sim 13.73$ $sec$.

\noindent{\bf Effect of the Size, $\boldsymbol{M}$, of Query Result Set:}  Figure \ref{fig:subfigures}(e) shows the Top\textit{M}-RICS query performance with different sizes, $M$, of the returned query result set, where $M$ varies from 1 to 5, and default values are used for other parameters. From the figure, larger $M$ indicates more community answers will be retrieved and refined. Thus, larger $M$ leads to higher time cost. Nonetheless, for different $M$ values, the time costs of Top\textit{M}-RICS remain low (i.e., $4.92 \sim 11.25$ $sec$).

\noindent{\bf Effect of the Number, $\boldsymbol{d}$, of Pivots:} Figure \ref{fig:subfigures}(f) shows the Top\textit{M}-RICS query performance for various numbers of pivots, where $d$ = 3, 4, 5, 6, and 8, and default values are used for other parameters.
When $d$ increases, the distance lower bounds from candidate communities $S_l$ to target community $Q$ are tighter, which incurs better searching order of candidate communities and achieves higher influence threshold earlier (or lower query costs). However, more pivots will also lead to higher computation costs for distances with lower bounds. Therefore, in the figure, for larger $d$ values, the wall clock time first decreases and then increases. Nonetheless, the wall clock times remain low (i.e., $4.92 \sim 9.97$ $sec$).

\noindent{\bf Effect of the Maximum Size, $\boldsymbol{N}$, of Seed Communities:} Figure \ref{fig:subfigures}(g) evaluates the performance of our Top\textit{M}-RICS approach, where the maximum size, $N$, of seed communities varies from 5 to 20, and other parameters are by default. The smaller $N$ is, while we have fewer candidate communities, the computational cost of the $k$-truss subgraph with maximum influence performed to obtain these candidate communities is greatly increased. Therefore, in the figure, when $N$ increases, the wall clock time decreases for all the three synthetic graphs. Nevertheless, the time costs remain low (i.e., $4.30 \sim 14.62$ $sec$) for different $N$ values.

\noindent{\bf Effect of the Size, $\boldsymbol{|V(G)|}$, of the Data Graph $G$: } Figure \ref{fig:subfigures}(h) tests the scalability of our Top\textit{M}-RICS approach, where graph size $|V(G)|$ = $10K$, $25K$, $50K$, $100K$, and $250K$, and the rest of parameters are set by their default values. From the figure, we can see that, with the increase of the graph size $|V(G)|$, the number of candidate seed communities also increases, which leads to higher pruning/refinement costs and more wall clock times. Nonetheless, even when $|V(G)| = 250K$ (i.e., $250K$ vertices in graph $G$), the time costs are less than 46.10 $sec$ for all the three synthetic graphs, which confirms the efficiency and scalability of our proposed Top\textit{M}-RICS approach on large-scale social networks.

\noindent{\bf Effect of Keyword Domain Size $\boldsymbol{|\Sigma|}$:} Figure \ref{fig:sigma} illustrates the Top\textit{M}-RICS query performance with different keyword domain sizes $|\Sigma|$ = 10, 20, 50, and 80, where other parameters are set to default values. 
From this figure, we can find that, since larger $\Sigma$ will improve the pruning power of keyword pruning, the community computational cost decreases. On the other hand, fewer candidate communities also lead to lower impact thresholds and lower pruning power. Thus, for all three synthetic graphs, the wall clock time decreases and then increases as $\Sigma$ increases. Nevertheless, the wall clock times remain low (i.e., $4.92 \sim 30.75$ $sec$).

\subsection{Ablation Study} 
To evaluate the effectiveness of our proposed pruning strategies, we conduct an ablation study over real/synthetic graphs. Figures \ref{fig:ablation} and \ref{fig:ablation_time} test the number of pruned communities and wall clock time of our Top$M$-RICS approach, respectively, for different pruning combinations by adding one more pruning strategy each time: (1) \textit{keyword pruning} only, (2) \textit{keyword + support pruning}, and (3) \textit{keyword + support + influence score pruning}, where all parameters are set to default values. From the figures, we can see that, by using more pruning strategies, the number of pruned communities increases by 1-3 orders of magnitude, and the time cost decreases by an order of magnitude accordingly.

\begin{figure}
    \centering
    \subfigure[pruning power] {
        \includegraphics[height=3cm]{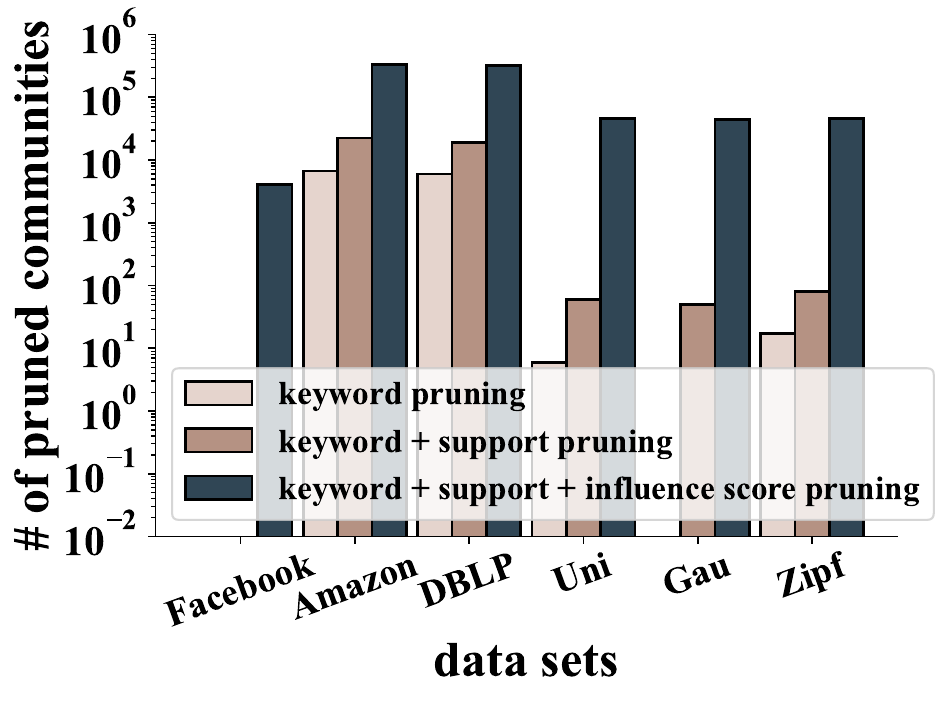}
        \label{fig:ablation}
    }\hspace{-0.3cm}
    \subfigure[time cost]{
        \includegraphics[height=3cm]{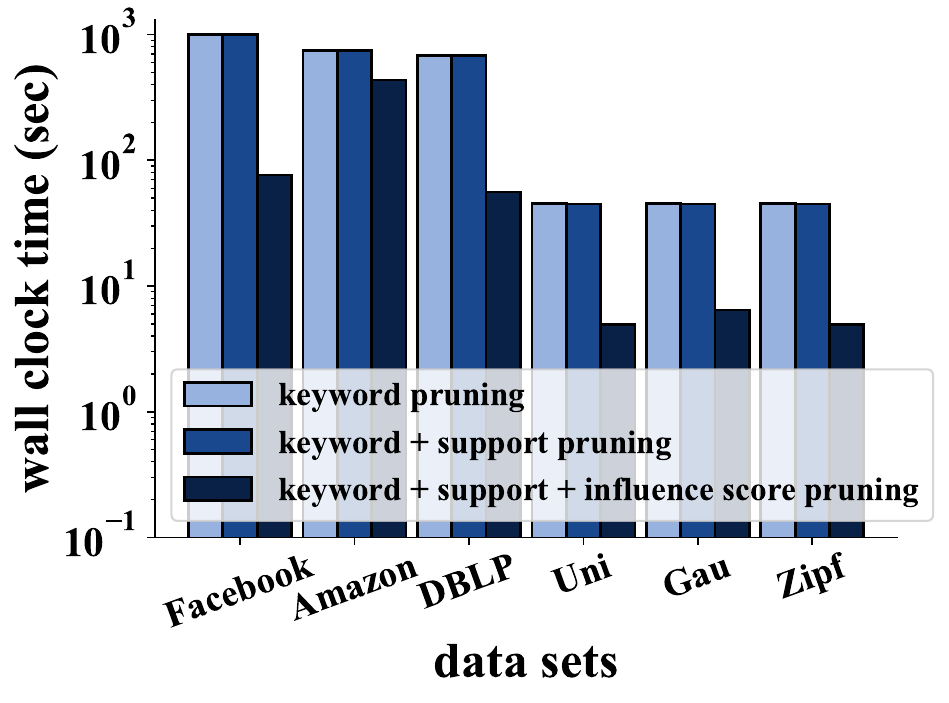}
        \label{fig:ablation_time}
    }
    \vspace{-1.5ex}
    \caption{ The ablation study of the Top\textit{M}-RICS performance.}
    \label{fig:ablation_all}
\end{figure}

{\color{black}
\subsection{Case Study} To evaluate the usefulness of our Top\textit{M}-RICS results, we conduct a case study to compare the influences of the Top\textit{M}-RICS seed community ($M=1$) under $k$-truss semantics with that of using $k$-core \cite{li2018persistent} over $DBLP$ data. 
Figure~\ref{fig:casestudy} visualizes the influence propagation between seed and target communities, where blue triangles and red stars are vertices in the seed and target communities, respectively, and shades of the edge color reflect influence scores.
From the figures, we can see that, for the same target community, although the 4-core community has more vertices, our Top\textit{M}-RICS seed community has an influence score of 15.74, which is significantly greater than that of the 4-core community (i.e., 4.72). This confirms the usefulness of community semantics in our Top\textit{M}-RICS problem, which can achieve high influences.

}
\begin{figure}[t]
    \subfigure{
        \includegraphics[width=0.48\textwidth]{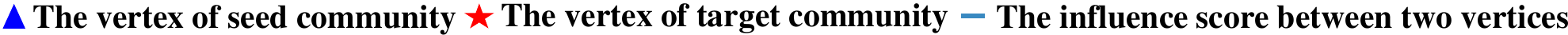}
        \label{fig:case3}
    }
    \centering
    \subfigure[Top$M$-RICS (4-truss, \newline \text{\quad} influence score = 15.74)]{
        \setcounter{subfigure}{1}
        \includegraphics[height=3cm]{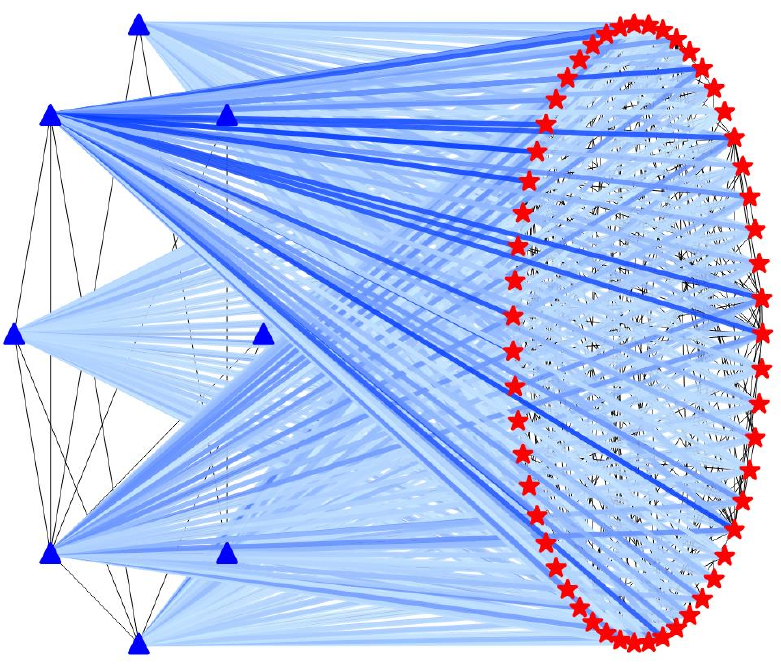}
        \label{fig:case1}}\hspace{-0.3cm}
    \subfigure[Top$M$-RICS (4-core, influence score \text{\qquad} = 4.72)]{
        \includegraphics[height=3cm]{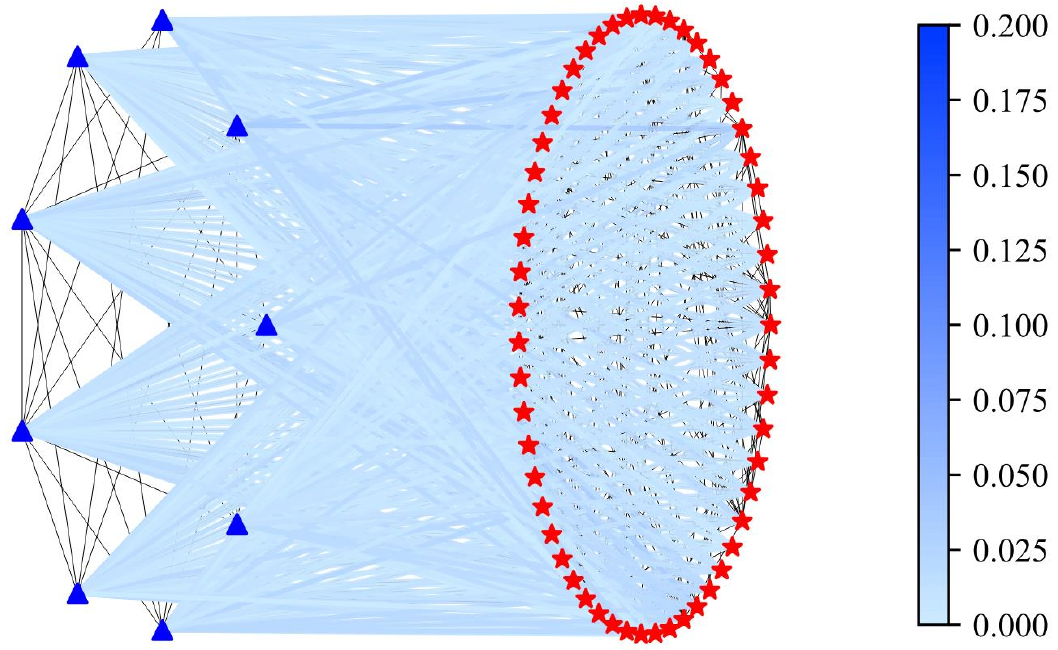}
        \label{fig:case2}
    }
    \vspace{-3ex}
    \caption{A case study of Top$M$-RICS with different community structure semantics over $DBLP$ ($M=1$).}
    \label{fig:casestudy}
    \vspace{-1ex}
\end{figure}

\subsection{Top\textit{M}-R$^2$ICS Performance Evaluation}

\noindent{\bf The Top\textit{M}-R$^2$ICS Performance on Real/Synthetic Graphs:} Figure~\ref{fig:R^2ICS+rs} illustrates the performance of our Top\textit{M}-R$^2$ICS approach on both real and synthetic graphs, compared with the Top\textit{M}-R$^2$ICS\_WoP and \textit{Optimal} methods, 
where all parameters are set by their default values in Table~\ref{tab:parameters}. From the figure, since Top\textit{M}-R$^2$ICS uses Lemma~\ref{lemma:R^2ICS influence} as an influence upper bound pruning strategy, it is unnecessary to specifically compute candidate vertices with a very small influence upper bound on the query target community. And, we can see the \textit{wall clock time} of Top\textit{M}-R$^2$ICS approach outperforms Top\textit{M}-R$^2$ICS\_WoP by about one order of magnitude and outperforms \textit{Optimal} by about two orders of magnitude. Moreover, every vertex is fully considered in our refinement process, and the accuracy of our method is 100\% as in \textit{Optimal}. These results confirm our overall method's effectiveness and our Top\textit{M}-R$^2$ICS's efficiency on real and synthetic graphs.

Figure~\ref{fig:R^2ICS+M} tests the effect of the size, $M$, of query answer set on the Top\textit{M}-R$^2$ICS query performance, by varying $M$ from 1 to 5, where other parameters are set to their default values. When $M$ increases, more candidate subgraphs need to be retrieved and refined, which incurs higher time cost. Nevertheless, the wall clock time remains low (i.e., $6.27 \sim 16.09$ $sec$).

\begin{figure}
    \centering
    \subfigure[time cost vs. graphs] {
        \includegraphics[height=3cm]{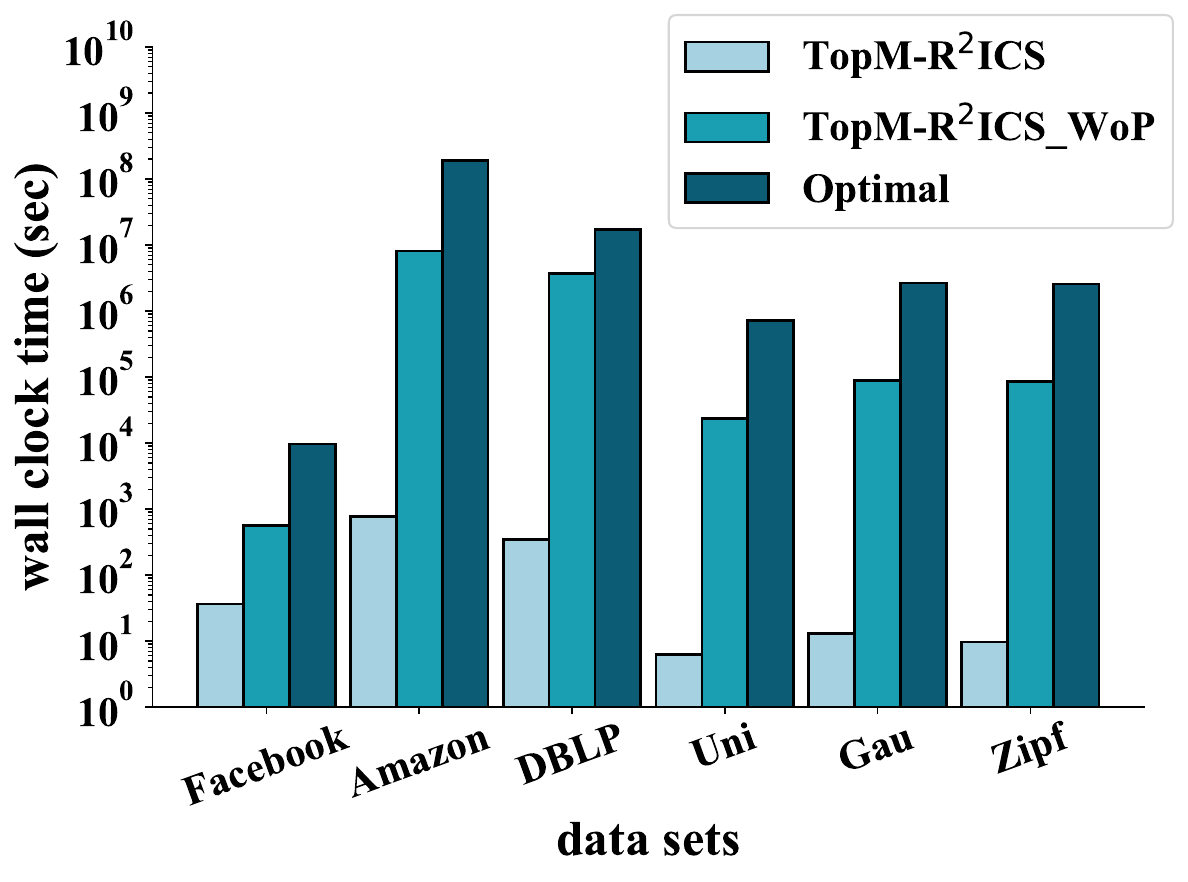}
        \label{fig:R^2ICS+rs}
    }\hspace{-0.3cm}
    \subfigure[query result size, $M$]{
        \includegraphics[height=3cm]{Figures/M.pdf}
        \label{fig:R^2ICS+M}
    }
    \vspace{-1.5ex}
    \caption{The Top\textit{M}-R$^2$ICS performance on real/synthetic graphs.}
    \label{fig:R^2ICS}
\end{figure}


\balance

\section{Related work}
\noindent\textbf{Community Search/Detection:} The \textit{community search} (CS) over social networks usually searches for connected subgraphs containing a given query vertex or a set of query vertices~\cite{fang2020survey, fang2020effective,sozio2010community}, whereas the \textit{community detection} (CD) detects all communities in social networks.
Prior works on CS \cite{cui2014local,bonchi2019distance,batagelj2003m,zhang2019unboundedness,cui2013online} adopted different cohesive subgraph models (e.g., $k$-core \cite{bonchi2019distance,batagelj2003m}, $k$-truss \cite{zhang2019unboundedness,cohen2008trusses}, and $k$-clique \cite{cui2013online,yuan2017index}) to find communities that contain a query vertex $q$.
Moreover, the foundation of many CD algorithms lies in the graph partitioning \cite{kernighan1970efficient, newman2013community} and clustering \cite{girvan2002community, clauset2004finding, blondel2008fast}, which focus on modularising dense structure parts in the data graph for the detection.
While most CS/CD studies considered retrieving communities whose own properties satisfy some constraints (e.g., query vertex, dense graph structure, and/or query keywords), our work focuses on reverse community search that obtains subgraphs most influencing a target community (instead of the entire graph).


\noindent\textbf{Influence Maximization:} The \textit{influence maximization} (IM) problem identifies a set of users as seed vertices with the maximum impact on other users within a given social network. Kempe et al. \cite{kempe2003maximizing} proposed two influence propagation models, \textit{Independent Cascade} (IC) and \textit{Linear Threshold} (LT) models, which have been widely used in the literature \cite{chen2010scalable,wang2012scalable,chen2015online,li2015real}. Chen et al. \cite{chen2010scalable} introduces the \textit{DegreeDiscount} heuristic algorithm for LT, presenting a scalable IM algorithm. \cite{wang2012scalable} proposed the PMIA heuristic algorithm for the IC model. However, such IM problems typically do not consider the constraints among seed vertices, whereas our Top\textit{M}-RICS problem identifies (densely connected) seed communities (containing query keywords) that can influence a target group.

\balance
\section{Conclusion}
\label{conclusion}
\textcolor{black}{
This paper proposed a novel Top\textit{M}-RICS problem, which returns top-\textit{M} seed communities with the maximum influences on a user-specified target community. Unlike existing works, the Top\textit{M}-RICS problem considers the influences of seed communities on a specific user group/community, rather than arbitrary users in social networks. To tackle the Top\textit{M}-RICS problem, we designed effective pruning strategies to filter out false alarms of candidate seed communities, and constructed an index to facilitate our proposed efficient Top\textit{M}-RICS query processing algorithm. 
We also formulated and tackled a variant of Top\textit{M}-RICS (i.e., Top\textit{M}-R$^2$ICS) with the relaxed structural constraints, by proposing effective pruning methods and an efficient query algorithm.
Extensive experiments on real/synthetic social networks validated the efficiency and effectiveness of our Top\textit{M}-RICS and Top\textit{M}-R$^2$ICS approaches.
}

\newpage

\bibliography{main}
\bibliographystyle{IEEEtran}

\end{document}